%% file: main.tex
\documentclass[11pt]{article}
\usepackage{graphicx}
\usepackage{latexsym}
\usepackage{amssymb}
\usepackage{amsmath}
\usepackage{amsthm}
\usepackage{bm}

\usepackage{mathtools}

\usepackage{nth}

\usepackage{forloop}
\usepackage[paper=letterpaper,margin=1in]{geometry}
\usepackage[ruled,vlined,linesnumbered]{algorithm2e}
\usepackage{times}
\usepackage{paralist}
\usepackage{nicefrac}
\usepackage{rotating}
\usepackage{tikz}
\usepackage{array,multirow}
\usetikzlibrary{positioning}
\usepackage{threeparttable}
\usepackage{hyperref}
\usepackage{etoolbox}

\newbool{shortVersion}
\newbool{SODACamera}
\newtheorem{theorem}{Theorem}[section]
\newtheorem{lemma}[theorem]{Lemma}

\newtheorem{corollary}[theorem]{Corollary}
\newtheorem{definition}[theorem]{Definition}
\newtheorem{proposition}[theorem]{Proposition}

\newtheorem{observation}[theorem]{Observation}

\makeatletter
\newtheorem*{rep@theorem}{\rep@title}
\newcommand{\newreptheorem}[2]{%
\newenvironment{rep#1}[1]{%
 \def\rep@title{#2 \ref{##1}}%
 \begin{rep@theorem}}%
 {\end{rep@theorem}}}
\makeatother

\newreptheorem{theorem}{Theorem}
\newreptheorem{reduction}{Reduction}
\newreptheorem{proposition}{Proposition}
\newreptheorem{lemma}{Lemma}

\DeclareMathOperator{\rank}{rank}
\DeclareMathOperator{\conv}{conv}

\DeclareMathOperator{\poly}{poly}
\DeclareMathOperator{\spn}{span}
\DeclareMathOperator{\argmax}{argmax}

\newcommand\E{\mathbb{E}}

\newcommand{\defcal}[1]{\expandafter\newcommand\csname c#1\endcsname{{\mathcal{#1}}}}
\newcommand{\defbb}[1]{\expandafter\newcommand\csname b#1\endcsname{{\mathbb{#1}}}}
\newcounter{calBbCounter}
\forLoop{1}{26}{calBbCounter}{
    \edef\letter{\Alph{calBbCounter}}
    \expandafter\defcal\letter
		\expandafter\defbb\letter
}

\newcommand{\rz}[1]{{\color{blue}RZ: #1}}

\newcommand{\eps}{{\varepsilon}}
\newcommand{\ie}{{\it i.e.}}
\newcommand{\eg}{{\it e.g.}}

\newcommand{\characteristic}{{\mathbf{1}}}

\title{Online Contention Resolution Schemes}

\author{Moran Feldman%
\thanks{School of Computer and Communication Sciences, EPFL. 
Email:
\href{mailto:moranfe3@gmail.com}{moranfe3@gmail.com}.
Supported by ERC Starting Grant 335288-OptApprox.}
\and
Ola Svensson\thanks{School of Computer and Communication Sciences, EPFL.
Email:
\href{mailto:ola.svensson@epfl.ch}{ola.svensson@epfl.ch}.
Supported by ERC Starting Grant 335288-OptApprox.}
\and
Rico Zenklusen%
\thanks{Department of Mathematics, ETH Zurich,
and Department of Applied Mathematics
and Statistics, Johns Hopkins University.
Email:
\href{mailto:ricoz@math.ethz.ch}{ricoz@math.ethz.ch}.
}%
}

\date{\today}

\begin{document}

\maketitle

%\begin{tikzpicture}[overlay, remember picture]
%\node[below right=0.5cm] at (current page.north west)
%  {\Large \ifbool{shortVersion}{Extended abstract}{Full version}};
%\end{tikzpicture}

\input{abstract.tex}


\thispagestyle{empty}

\pagenumbering{Alph}
\newpage

\pagenumbering{arabic}

%-----------------------------------------------------

\input{introduction}

\input{intResults}
\input{intApplications}

\input{contres}

\ifbool{shortVersion}{}{
\input{SubmodularResults}
\input{applications}

}% end ifbool{shortVersion}

%-----------------------------------------------------

%\section*{Acknowledgement}

\bibliographystyle{plain}
\bibliography{lit}

\ifbool{shortVersion}{}{
\appendix
\input{matroidConcepts}
}

\InputIfFileExists{ricoNotes.tex}{}{}

%\ifbool{shortVersion}{}{
%\appendix
%
%%% Include appendix sections with
%%% \input statements.
%
%}%end ifbool

\end{document}

%% file: abstract.tex
\begin{abstract}\ifbool{SODACamera}{\small\baselineskip=9pt}{} 
We introduce a new rounding technique designed for online optimization
problems, which is related to contention resolution schemes, a technique
initially introduced in the context of submodular function maximization.
Our rounding technique, which we call \emph{online contention
resolution schemes} (OCRSs), is applicable to many online
selection problems, including Bayesian online selection,
oblivious posted pricing mechanisms, and stochastic probing
models. It allows for handling a wide set of constraints,
and shares many strong properties of offline contention
resolution schemes. In particular, OCRSs for different
constraint families can be combined to obtain an OCRS
for their intersection. Moreover, we can approximately
maximize submodular functions in the online settings
we consider.

We, thus, get a broadly applicable framework for several online selection
problems, which improves on previous approaches in terms of the types of
constraints that can be handled, the objective functions that can be dealt
with, and the assumptions on the strength of the adversary.  Furthermore, we
resolve two open problems from the literature; namely, we present the first
constant-factor constrained oblivious posted price mechanism for matroid
constraints, and the first constant-factor algorithm for weighted stochastic
probing with deadlines.

\end{abstract}

%% file: introduction.tex
\section{Introduction}

Recently, interest has surged in Bayesian and stochastic online optimization
problems. These are problems where we can use limited a priori information to
select elements arriving online, often subject to 
classical combinatorial constraints such as matroids, matchings and
knapsack constraints.
Examples include posted pricing
mechanisms~\cite{chawla_2010_multi-parameter,yan_2010_mechanism,kleinberg_2012_matroid},
prophet inequalities~\cite{kleinberg_2012_matroid},
probing models~\cite{gupta_2013_stochastic,adamczyk_2014_submodular},
stochastic matchings~\cite{bansal_2012_lp},
and secretary
problems~\cite{babaioff_2007_matroids,lachish_2014_competitive,%
feldman_2015_simple}.\footnote{Strictly speaking, secretary problems have no a priori information. However, as items arrive in a random order, most algorithms first observe a fraction of the elements (serving as the a priori information), and then devise an online strategy based on this information.}
Simultaneous with this development, interest has arose also
in generalizing the optimization of linear objective functions to relevant
nonlinear objective
functions.
A particular focus was set on
%the online maximization of
submodular functions, which is a function class that
captures the property of diminishing returns, a very
natural property in many of the above-mentioned
settings~\cite{lehmann_2006_combinatorial,%
adamczyk_2014_submodular,bateni_2013_submodular,%
feldman_2015_submodular}. % QQQ: Paper of Moran and Rico is missing.

%Recent progress on these models revealed many
%new techniques showing how structural properties
%of the underlying constraints can be leveraged
%in online settings.

A very successful approach  for these problems is based on first using the
a priori information to formulate an (often linear) relaxation whose optimal
fractional solution $x^*$ upper bounds the performance of any online (or even
offline) algorithm. Then, $x^*$ is used to devise an online algorithm whose
goal is to recover a solution of a similar objective value as $x^*$.
%
%solving a (often linear) relaxation
%and transforming the optimal fractional solution $x^*$, which
%was obtained offline, into an online algorithm with
%similar expected objective value.
%
Such an online algorithm can also be
interpreted as an online rounding procedure
for rounding $x^*$.
In particular, online rounding approaches have
recently been used to obtain nearly optimal 
and surprisingly elegant results for
stochastic matchings
(see Bansal et al.~\cite{bansal_2012_lp}),
and for a very general probing
model introduced by Gupta
and Nagarajan~\cite{gupta_2013_stochastic}
with applications in posted pricing mechanisms, online
matching problems and beyond.

A key ingredient in the general rounding algorithms 
presented in~\cite{gupta_2013_stochastic} are so-called
contention resolution schemes (CRSs), a rounding technique
introduced by Chekuri et al.~\cite{chekuri_2014_submodular}
in the context of (offline) submodular function
maximization.
CRSs  are defined with respect
to a constraint family, like matroids, matching, or
knapsack constraints.
Interestingly, the existence of a so-called \emph{ordered}
CRS for the given constraint family
is all that is needed to apply the techniques
of~\cite{gupta_2013_stochastic}.
Whereas this generality is very appealing, there are some
inherent barriers in current
CRSs that hinder a broader
applicability to online settings
beyond the probing model defined
in~\cite{gupta_2013_stochastic}.
More precisely, most settings considered
in~\cite{gupta_2013_stochastic} require that
the \emph{algorithm can choose the order}
in which to obtain new online information about an
underlying ground set over which the objective is optimized.
This is due to the fact that most CRSs 
need to round the components of a fractional
point $x^*$ step by step in a particular order.

In this paper we introduce a stronger
notion of contention resolution schemes that overcomes this restriction and allows for the online
information to  be presented \emph{adversarially}. 
We show that such schemes exist for many interesting
constraint families, including matroid constraints
and knapsack constraints.
As we discuss in Section~\ref{subsec:app}, this leads to a broadly applicable
online rounding framework that works in considerably more
general settings than previous approaches. Furthermore,
our techniques answer two open problems from the literature:
we show the existence of a constrained oblivious
posted-pricing mechanism (COPM)
for matroids, intersection of matroids, and further constraints
families (a question first
raised in~\cite{chawla_2010_multi-parameter});
and we get an $O(1)$-competitive algorithm
for the weighted probing problem with deadlines introduced
in~\cite{gupta_2013_stochastic}.
Additionally, our rounding approach yields optimal guarantees (up to moderate constant
factors) for a class of online submodular function maximization problems.

%Recently, Gupta and Nagarajan~\cite{gupta_2013_stochastic}
%presented a very general framework, he so-called
%\emph{probing model}, that allows for capturing a large set
%of such online rounding problems linked to online matchings 
%and sequential posted pricing mechanisms.

\medskip
\iffalse

\rz{
Below, I briefly introduced the prophet inequality
setting as an example for further notions and ideas.
\begin{itemize}
\item One could try to further shorten this example. Or
start right-away with the actual setting we have where
we have a point in the matroid polytope etc., and only
later link it to an application.

\item Also, the problem in its form below is already solved.

\item On the other hand, it allows for nicely showing certain
ideas and how settings of interest can arise. 
\end{itemize}
}
\fi

Before we formally define our rounding framework (Section~\ref{subsec:ocrs}), state our results (Section~\ref{subsec:results}) and describe the aforementioned applications (Section~\ref{subsec:app}), it is helpful to introduce our rounding framework in the light of a concrete example.
%For the sake of exposition, and to have a concrete example
%in mind when we introduce our rounding framework,
Consider the following Bayesian online selection problem studied
by Kleinberg and Weinberg~\cite{kleinberg_2012_matroid}
in the context of prophet inequalities.
There is a finite set $N$ of items or elements,
and a nonnegative 
random variable $Z_e$ for each $e\in N$,
where all $\{Z_e\}_{e\in N}$ are independent.
The distributions of all $Z_e$ are known and for
simplicity we assume they are continuous.
Furthermore, a matroid $M=(N,\mathcal{F})$ on
$N$ is given.\footnote{We recall that
a matroid $M=(N,\mathcal{F})$ consists of a finite
ground set $N$ and a nonempty family
$\mathcal{F}\subseteq 2^N$ of subsets of $N$
satisfying: (i) If $I\in \mathcal{F}, J\subseteq I$,
then $J\in \mathcal{F}$, and
(ii) if $I,J\in \mathcal{F}, |J|>|I|$, then
$\exists e\in J\setminus I$ with $I\cup\{e\}\in \mathcal{F}$.
If not stated otherwise, we assume that matroids are
given by an independence oracle that, for every
$I\subseteq N$, returns whether $I\in \mathcal{F}$.}
Let $\{z_e\}_{e\in N}$ be realizations of the random
variables $\{Z_e\}_{e\in N}$.
The goal is to select a subset $I\subseteq N$
of the elements that is independent, \ie,
$I\in \mathcal{F}$, and whose value $z(I):=\sum_{e\in E}z_e$
is as large as possible.
The way how elements can be selected works as follows.
Elements $e\in N$ reveal one by one their realization $z_e$,
in a fixed prespecified order that is unknown to the algorithm.
Whenever a value $z_e$ is revealed, one has to choose whether
to select $e$ or discard it, before the next element reveals
its realization.

A natural way to approach this problem is to define a threshold
$t_e \geq 0$ for each $e\in N$ and only accept elements
$e\in N$ whose realization is at least the threshold, \ie,
$z_e \geq t_e$; we call such elements \emph{active}.
Let $x_e$ be the probability of $e$ being active, \ie,
$x_e = \Pr[Z_e \geq t_e]$.
Notice that the set of all active elements is distributed
like a random set that contains each element $e$ independently
with probability $x_e$. We denote such a set by $R(x)$.
\ifbool{shortVersion}{As we discuss in the long version
of the paper}{As we show in Section~\ref{sec:app}},
using a convex relaxation one can find thresholds
$t_e$ such that: (i) $x=(x_e)_{e\in N}$ is in the
matroid polytope $P_\mathcal{F}$,\footnote{The matroid polytope
  $P_\mathcal{F}$ is the convex hull of all characteristic vectors
of independent sets. In particular
it can be described by
$P_\mathcal{F}=\{x\in \mathbb{R}_{\geq 0}^N \mid x(S) \leq \rank(S)
\;\forall S\subseteq N\}$, where
$\rank(S)=\max\{|I|\mid I\subseteq S, I\in \mathcal{F}\}$
is the \emph{rank function} of $M$.
} and
(ii) an algorithm that disregards the matroid constraint
and accepts any active element would have an expected
return at least as good as the one of an optimal offline
algorithm.

Our goal is to design an online algorithm that only selects
active elements, such that an independent set $I$ is
obtained where $\Pr[e\in I] \geq c\cdot x_e$
for all $e\in E$,
where $c\in (0,1]$ is a constant
as large as possible. It is not hard to check
that such a procedure would lead to an objective
value of at least $c$ times the offline optimum.
The guarantee we are seeking
closely resembles the notion of $c$-balanced
CRSs as defined in~\cite{chekuri_2014_submodular},
which is an \emph{offline} algorithm that
depends on $x$ and returns for any set
$S\subseteq N$ a (potentially random)
subset $\pi(S)\subseteq S$ with $\pi(S)\in \mathcal{F}$ 
such that $\Pr[e\in \pi(R(x))]\geq c\cdot x_e$.
The only reason why this procedure is not applicable
in the above context is that, in general, $\pi$ needs
to know the realization of the full set $R(x)$ in advance
to determine $\pi(R(x))$.
However, $R(x)$ is revealed element by element in the
above selection problem.
A key observation in~\cite{gupta_2013_stochastic} is that
some CRSs do not need to know the full set $R(x)$
upfront, but can round step by step if the elements come
in some prescribed order chosen by the algorithm.
However, in the above setting, as in many other 
combinatorial online problems, the order cannot be chosen
freely.

We overcome this limitation through a considerably
stronger notion of CRSs, which we
call \emph{online contention resolution schemes}
(OCRSs).

\subsection{Online contention resolution schemes\ifbool{SODACamera}{.}{}}
\label{subsec:ocrs}

OCRSs, like classical contention resolution schemes,
are defined with respect to a relaxation of the
feasible sets of a combinatorial optimization problem.
Consider a finite ground set $N=\{e_1,\dots, e_n\}$,
and a family of
\emph{feasible} subsets $\mathcal{F}\subseteq 2^N$,
which is down-closed, \ie, if $I\in \mathcal{F}$
and $J\subseteq I$ then $J\in \mathcal{F}$.
Let $P_{\mathcal{F}}\subseteq [0,1]^N$
be the polytope corresponding
to the feasible sets $\mathcal{F}$, \ie, $P_{\mathcal{F}}$
is the convex hull of all characteristic vectors
of feasible sets:
\begin{equation*}
P_{\mathcal{F}} = \conv(\{\characteristic_F \mid I\in \mathcal{F}\})
\enspace.
\end{equation*}
We highlight that throughout this paper we focus on
down-closed feasibility constraints.

\begin{definition}[relaxation]
We say that a polytope $P\subseteq [0,1]^N$ is a relaxation
of $P_{\mathcal{F}}$ if it contains the
same $\{0,1\}$-points, i.e.,
$P\cap \{0,1\}^N = P_{\mathcal{F}} \cap \{0,1\}^N$.
\end{definition}

%\begin{remark}
%More generally, we could have defined a relaxation without
%requiring that it is a polytope. However, since we only
%use polytopes as relaxations in the following, we require
%a relaxation to be a polytope for simplicity and to avoid
%confusion.
%
%Due to the same reasons, we restricted ourselves to
%combinatorial problems whose feasible solutions can
%be described as subsets of a ground set. The above
%notions and also the notion of OCRS can be extended
%to other problem settings like integer programs.
%\end{remark}

We start by defining online contention resolution schemes (OCRS)
simply as algorithms that can be applied to the online
selection problem highlighted in the introduction.
The performance of  an OCRS is then characterized by additional
properties that we define later.

\begin{definition}[\ifbool{SODACamera}{OCRS}{Online contention resolutions scheme (OCRS)}]
Let us consider the following online selection setting. A
point $x\in P$ is given and let $R(x)$ be a random subset
of \emph{active elements}.
The elements $e\in N$ reveal one by one
whether they are active, i.e., $e\in R(x)$, and the decision whether to select an active element is taken irrevocably before the next element is revealed. 
An OCRS for $P$ is an online algorithm that
selects a subset $I\subseteq R(x)$ such that
$\characteristic_I \in P$.
\end{definition}

Most of the OCRSs that we present follow a common
algorithmic theme, which leads us to the definition
of \emph{greedy OCRS}.

\begin{definition}[Greedy OCRS]
Let $P\subseteq [0,1]^N$ be a relaxation for the
feasible sets $\mathcal{F}\subseteq 2^N$.
A greedy OCRS $\pi$ for $P$ 
is an OCRS that for any $x\in P$
defines a   down-closed subfamily of feasible sets
$ \mathcal{F}_{\pi, x} \subseteq \mathcal{F}$,
and an element $e$ is selected when it arrives 
if, together with the already selected elements,
the obtained set is in $\mathcal{F}_{\pi,x}$.

If the choice of $\mathcal{F}_{\pi, x}$ given $x$
is randomized, we talk about a \emph{randomized} greedy OCRS;
otherwise, we talk about a \emph{deterministic}
greedy OCRS. We also simplify notation and abbreviate $\mathcal{F}_{\pi, x}$ by
$\mathcal{F}_x$ when the OCRS $\pi$ is clear from the context.
\end{definition}

For simplicity of presentation, and because all
our main results are based on greedy OCRSs, we
restrict our attention to this class of OCRSs,
and focus on greedy OCRSs when defining and
analyzing properties.

As mentioned in the example shown in the introduction,
a desirable property of OCRSs would be that every
element $e\in N$ gets selected with probability
at least $c\cdot x_e$ for a constant $c>0$ as large
as possible. This property is called
\emph{$c$-balancedness} in the context of classical
contention resolution schemes.
However, to be precise about such properties in the
online context that we consider, one has to specify
the power of the adversary who chooses the order
of the elements.
Adversaries of different strengths have been
considered in various online settings.
For example, one arguably weak type of
adversary is an \emph{offline adversary},
who has to choose the order of the elements upfront,
before any elements get revealed.
On the other end, the most powerful adversary
that can be considered is what we call the
\emph{almighty adversary}; an almighty adversary
knows upfront the outcomes of all random events, which includes the realization
of $R(x)$ and  the outcome of the random bits that the algorithm may query. An
almighty adversary can thus calculate exactly how the algorithm will behave  and reveal the elements in a worst case order.
%random choices that the algorithm may make.
%
A typical adversary type that is in between these
two extremes is an \emph{online adversary}, who
can choose the next element to reveal online depending on
what happened so far; thus, it has the same information
as the online algorithm.
Throughout this paper, when not indicated otherwise,
we assume to play against the almighty adversary.

In the context of greedy OCRSs, we define a considerably
stronger notion than $c$-balancedness, which we call
\emph{$c$-selectability}, and which leads to results
even against the almighty adversary.
In words, a greedy OCRS is $c$-selectable if with probability
at least $c$, the random set $R(x)$ is such that an element
$e$ is selected no matter what other elements $I$ of
$R(x)$ have been selected so far as long
as $I\in \mathcal{F}_x$. Thus, it guarantees that an element is selected with
probability at least $c$ against any (even the almighty) adversary.

\begin{definition}[$c$-selectability]\label{def:cSelectable}
Let $c\in [0,1]$.
A greedy OCRS for $P$ is $c$-selectable if for
any $x\in P$ we have 
\begin{equation*}
\Pr[ I \cup \{e\} \in \mathcal{F}_x \;\;\;
\forall \;I\subseteq R(x),
I\in \mathcal{F}_x]
\geq c \ifbool{SODACamera}{\quad}{\qquad} \forall e\in N
\enspace.
\end{equation*}
\end{definition}
We highlight that the probability in
Definition~\ref{def:cSelectable} is over the random
outcomes of $R(x)$ when dealing with a deterministic
greedy OCRS; when the greedy OCRS is randomized,
then the probability is over $R(x)$ and the random
choice of $\mathcal{F}_x$.
We call an element $e\in N$ \emph{selectable}
for a particular realization of $R(x)$ and random
choice of $\mathcal{F}_x$ if $I\cup\{e\}\in \mathcal{F}_x$
for all $I\subseteq R(x)$ with $I \in \mathcal{F}_x$.

As aforementioned, the $c$-selectability is a very strong property
that implies guarantees against any adversary.
Despite this strong definition, we show
that $\Omega(1)$-selectable greedy OCRSs exist for
many natural constraints.

\medskip

Often, a larger factor $c$ can be achieved when $x$ is
supposed to be in a down-scaled version of $P$.
This is similar to the situation in classical contention resolution
schemes.

\begin{definition}[$(b,c)$-selectability]
Let $b,c \in [0,1]$.
A greedy OCRS for $P$ is $(b,c)$-selectable if for
any $x\in b\cdot P$ we have 
\begin{equation*}
\Pr[ I \cup \{e\} \in \mathcal{F}_x \;\;\;
\forall \;I\subseteq R(x),
I\in \mathcal{F}_x]
\geq c \ifbool{SODACamera}{\quad}{\qquad} \forall e\in N
\enspace.
\end{equation*}
\end{definition}

Notice that a $(b,c)$-selectable greedy OCRS implies a randomized $bc$-selectable greedy OCRS because we can ``scale down'' $x$ online by only considering each element $e$ with probability $b$ independent of the other elements. 
\begin{observation} \label{obs:pair_to_single_selectability}
  A $(b,c)$-selectable greedy OCRS for $P$ implies a (randomized) $bc$-selectable greedy OCRS for $P$. 
\end{observation}

The existence of OCRSs is interesting even regardless of efficiency issues. Still, in many applications it is important to have efficient OCRSs.
\begin{definition}[efficiency]
A greedy OCRS $\pi$ is \emph{efficient} if there exists a polynomial
time algorithm that, for a given input $x$,  samples an efficient independence oracle for the set
$\cF_{\pi,x}$. That is, an oracle that answers in polynomial time queries of the form: is a set $S\subseteq N$  in $\mathcal{F}_{\pi,x}$?
%the oracle can in polynomial time answer whether a given input is in $\cF_{\pi,x}$.
%second
%polynomial time algorithm that verifies whether a given input set is in the
%sampled $\cF_x$.
\end{definition}

We next summarize our technical results before highlighting the
implications of our results to various
online settings.

%% file: intResults.tex
\subsection{Our results\ifbool{SODACamera}{.}{}}\label{subsec:results}

Our first technical result proves the existence of greedy OCRSs with constant
selectability for relaxations of several interesting families of constraints.
All the greedy OCRSs described by Theorem~\ref{thm:direct_OCRS} are either
efficient, or can be made efficient at the cost of an arbitrarily small
constant $\eps > 0$ loss in the selectability guarantee. 

\begin{theorem} \label{thm:direct_OCRS}
There exist:
\vspace{-0.5em}
\begin{itemize}
\setlength\itemsep{0.0em}
	\item For every $b \in [0, 1]$, a $(b, 1 - b)$-selectable deterministic greedy OCRS for matroid polytopes.
	\item For every $b \in [0, 1]$, a $(b, e^{-2b})$-selectable randomized greedy OCRS for matching polytopes.\footnote{Our greedy OCRS works also for a weaker relaxation of matching which only bounds the degree of each node.}
	\item For every $b \in [0, \nicefrac{1}{2}]$, a $(b, \frac{1 - 2b}{2 - 2b})$-selectable randomized greedy OCRS for the natural relaxation of a knapsack constraint.
\end{itemize}
\end{theorem}

Interestingly, it turns out that there is no $(b, c)$-selectable
\emph{deterministic} greedy OCRS for the natural relaxation of a knapsack
constraint for any constants $b$ and $c$. This stands in contrast to the
case of the matching polytope, for which the randomized greedy OCRS given by
Theorem~\ref{thm:direct_OCRS} can be made deterministic at the cost of only
a small loss in the selectability.

Like offline CRSs, greedy OCRSs can be combined to form greedy OCRSs for more involved constraints.

\begin{theorem}\label{thm:combineOCRSs}
If $\pi^1$ is a $(b,c_1)$-selectable greedy OCRS for a polytope $P_1$, and $\pi^2$ is a $(b, c_2)$-selectable greedy OCRS for a polytope $P_2$, then there exists a $(b, c_1 \cdot c_2)$-selectable greedy OCRS for the polytope $P_1 \cap P_2$. Moreover, the last greedy OCRS is efficient if $\pi^1$ and $\pi^2$ are.
\end{theorem}

Notice that Theorem~\ref{thm:combineOCRSs} can be applied repeatedly to combine several OCRSs. Thus, Theorems~\ref{thm:direct_OCRS} and~\ref{thm:combineOCRSs} prove together the existence of constant selectability greedy OCRSs for any constant intersection of matroid, matching and knapsack constraints.

It is easy to see that, given a non-negative increasing linear objective function, a $(b, c)$-selectable greedy OCRS for a polytope $P$ can be used to round online a vector $x \in bP$ while losing only a factor of $c$ in the objective. Theorem~\ref{thm:submodular_basic} proves this result formally, and extends it to nonnegative submodular\footnote{A set function $f\colon 2^N \to \bR$ is submodular if $f(A) + f(B) \geq f(A \cup B) + f(A \cap B)$ for every two sets $A, B \subseteq N$.} functions. To state this theorem, we need to define some notation. Given a function $f\colon 2^N \to \bR$, the \emph{multilinear} extension of $f$ is a function $F\colon [0,1]^N \to \bR$ whose value for a vector $x \in [0, 1]^N$ is $F(x) = \bE[f(R(x))]$. Informally, $F(x)$ is the expected value of $f$ over a set obtained by randomly rounding every coordinate of $x$ independently.

\begin{theorem} \label{thm:submodular_basic}
Given a nonnegative monotone\footnote{A set function $f\colon 2^N \to \bR$ is
monotone if $f(A) \leq f(B)$ for every two sets $A \subseteq B \subseteq N$.}
submodular function $f\colon 2^N \to \bR_{\geq 0}$ and a $(b, c)$-selectable
greedy OCRS for a polytope $P$, applying the greedy OCRS to an input $x \in bP$
results in a random set $S$ satisfying $\bE[f(S)]
\geq c \cdot F(x)$, where $F$ is the multilinear extension of $f$. Moreover,
even if $f$ is not monotone, $\bE[f(R(\nicefrac{1}{2} \cdot \characteristic_S))] \geq (c/4) \cdot F(x)$,
where the random decisions used to calculate $R(\nicefrac{1}{2} \cdot
\characteristic_S)$ are considered part of the algorithm, and thus, known to
the almighty adversary.
\end{theorem}

In many applications the use of Theorem~\ref{thm:submodular_basic} requires finding \emph{offline}, using the available a priori information, a vector $x$ (approximately) maximizing the multilinear extension $F$. This can often be done using known algorithms. For example, C{\u{a}}linescu et al.~\cite{culinescu_2011_maximizing} proved that given a non-negative monotone submodular function $f\colon 2^N \to \mathbb{R}_{\geq 0}$ and a solvable\footnote{A polytope is \emph{solvable} if one can optimize linear functions over it.} polytope $P \subseteq [0, 1]^N$, one can efficiently find a fractional point $x \in P$ for which $F(x) \geq (1 - e^{-1}) \cdot \max\{f(S) \mid \characteristic_S \in P\}$. Chekuri et al.~\cite{chekuri_2014_submodular} showed that even
when $f$ is not monotone, an analogous result can be obtained with 
a worse constant factor of $0.325$ instead of $1-e^{-1}$ when
$P$ is solvable and down-closed.
A simpler procedure with a stronger constant factor
was later presented by
Feldman et al.~\cite{feldman_2011_unified}, implying that one can efficiently find a fractional point $x \in P$ for which $F(x) \geq (e^{-1} - o(1)) \cdot \max\{f(S) \mid \characteristic_S \in P\}$ as long as the polytope $P$ is solvable and down-closed.

The result of Theorem~\ref{thm:submodular_basic} for a non-monotone submodular objective can sometimes be improved when assuming an online adversary (instead of an almighty one). The class of OCRSs for which this can be done is a bit involved to define,
\ifbool{shortVersion}{and we provide a precise definition in the
long version of the paper.}%
{and we defer its definition to \ifbool{SODACamera}{the full version of this paper~\cite{feldman_2015_online}}{Section~\ref{sec:selectability_to_approximation}}.}
We state here only the following special case of the result we prove.

\begin{theorem} \label{thm:non_negative_special}
  Let $\pi$ be a $(b,c)$-selectable greedy OCRS $\pi$ for a polytope $P$ that was obtained by using Theorem~\ref{thm:combineOCRSs} to combine the OCRS of Theorem~\ref{thm:direct_OCRS}. 
Then, for every given
non-negative submodular function $f\colon 2^N \to \bR_{\geq 0}$ there exists an
OCRS $\pi'$ for $P$ that for every input vector $x \in bP$ and online adversary selects
a random set $S$ such that $\bE[f(S))] \geq
c \cdot F(x)$.
\end{theorem}

The OCRS $\pi'$ guaranteed by Theorem~\ref{thm:non_negative_special} is not efficient. However, if $\pi$ is efficient then $\pi'$ can be made efficient at the cost of an additive loss of $|N|^{-d} \cdot \max\{f(\{e\}) \mid x_e > 0\}$ in the guarantee (where $d$ is any positive constant). %Our final result studies the possibility of getting deterministic greedy OCRSs for the polytopes studied by Theorem~\ref{thm:direct_OCRS}. Interestingly, we show that randomization is essential for the natural relaxation of knapsack constraints.

%\begin{theorem} \label{thm:deterministic}
%There exists:
%\begin{itemize}
	%\item for every $b \in [0, 1]$, a $(b, (1 - b)^2)$-selectable deterministic greedy OCRS for matching polytopes.
	%\item for every $n \geq 1$, a knapsack constraint over a ground set of $n$ elements such that no deterministic greedy OCRS for the natural relaxation of this constraint is $(b, c)$-selectable for any pair of $b \in [0, 1]$ and $c > (1-b)^{n-1}$.
%\end{itemize}
%\end{theorem}

%% file: intApplications.tex
\subsection{Applications\ifbool{SODACamera}{.}{}}\label{subsec:app}

In this section we present a few applications for our technical results. All these applications were previously studied in the literature, and connections have been found between them. In this work we show that all three applications can be reduced to finding appropriate OCRSs. In addition to proving new results, we believe that these reductions into one common setting clarify the connections between the three applications.

\subsubsection*{Prophet inequalities for
Bayesian online selection problems\ifbool{SODACamera}{.}{}}

Consider again the Bayesian online selection 
problem we sketched earlier in the introduction.
We recall that the setting in this problem consists of a matroid $M=(N,\mathcal{F})$
and independent non-negative random variables $Z_e$
for every $e\in N$ with known distributions.
Moreover, the random variables $Z_e$ satisfy
$\max_{e\in N} \E[Z_e]< \infty$.
An offline adversary chooses upfront the order
in which the elements $e\in N$ reveal a realization
$z_e$ of $Z_e$.
The task is to select online an independent
set of elements $I\in\mathcal{F}$ with 
total weight $z(I)=\sum_{e\in I} z_e$ as high
as possible.

A fundamental result about the relative power
of offline and online algorithms in a Bayesian setting
was obtained by Krengel, Sucheston and Garling
(see~\cite{krengel_1978_semiamarts}) for the
special case when $M$ is the uniform matroid
of rank one, \ie, precisely one element can be selected.
They showed that there exists a selection algorithm
returning a single element of expected weight
as least
$\frac{1}{2}\E[\max_{e\in N} Z_e]$, \ie, half of
the weight of the best offline solution, which is
the best solution obtainable
by an algorithm that knows all realizations upfront.
Recently, Kleinberg and
Weinberg~\cite{kleinberg_2012_matroid} extended this
result considerably by showing that the
same guarantee can be obtained when selecting multiple
elements that have to be independent in the matroid $M$,
\ie, there exists an online algorithm returning a
set $I\in \mathcal{F}$ satisfying
\begin{equation}\label{eq:prophetMatroid}
\E\left[\sum_{e\in I} Z_e\right] \geq
\frac{1}{2} \E\left[ \max\left\{\sum_{e\in S} Z_e
\;\middle\vert\; S\in \mathcal{F}\right\} \right]
\enspace.
\end{equation}
Inequalities of type~\eqref{eq:prophetMatroid} are
often called \emph{prophet inequalities} due to the
interpretation of the offline adversary as a prophet.
Moreover, Kleinberg and Weinberg generalized their result
to the setting where $\mathcal{F}$ are the common
independent sets in the intersection of $p$ matroids.
For this setting, they present an online algorithm
whose expected profit is at least $\frac{1}{4p-2}$
times the expected maximum weight of a feasible set.
Kleinberg and Weinberg's algorithms work not just
against an offline adversary, which is the adversary
type typically assumed in Bayesian online selection,
but also against an online adversary.
%This fact is
%important for an application of their work to Bayesian
%mechanism design.

Using a simple, yet very general link between
greedy OCRSs and prophet inequalities we can generate
prophet inequalities from greedy OCRSs.

\begin{theorem}\label{thm:prophet}
Let $\mathcal{F}\subseteq 2^N$ be a down-closed
set family and $P$ be a relaxation of $\mathcal{F}$.
If there exists a $c$-selectable greedy OCRS for $P$
then there is an online algorithm for the Bayesian
online selection problem with almighty adversary
that returns a set $I\in \mathcal{F}$ satisfying
\begin{equation*}
\E\left[\sum_{e\in I} Z_e\right] \geq
c \cdot \E\left[ \max\left\{\sum_{e\in S} Z_e
\;\middle\vert\; S\in \mathcal{F}\right\} \right]
\enspace.
\end{equation*}
\end{theorem}

\ifbool{shortVersion}{As we discuss in the long
version of the paper,}%
{As we discuss in Section~\ref{sec:app},}
the above theorem
can be made constructive in many cases, assuming that
the OCRS is efficient and some natural optimization problems
involving the distributions of the random weights $Z_e$ can
be solved efficiently.

Our results show that constant-factor prophet inequalities
are often possible even against an almighty adversary.
Moreover, we get $\Theta(1)$-factor prophet inequalities for a wide
set of new constraint families.

\begin{corollary}
There are $\Theta(1)$-factor prophet inequalities 
for the Bayesian online selection problem against
the almighty adversary for any constraint
family that is an intersection
of a constant number of matroid, knapsack, and
matching constraints.
\end{corollary}

In contrast, so far, the most general prophet
inequality was the $\frac{1}{4p-2}$-factor
prophet inequality of Kleinberg and Weinberg for the
intersection of $p$ matroids.
Interestingly, even for this specific setting of
the intersection of $p$ matroids, considered
by Kleinberg and Weinberg, our general approach allows
for obtaining a better constant for $p\geq 4$ (against
a stronger adversary).

\begin{corollary}
There is an $\frac{1}{e(p+1)}$-factor prophet inequality
for the Bayesian online selection problem against
the almighty adversary when the feasible sets are
described by the intersection of $p$ matroids.
\end{corollary}

That our results hold against an almighty adversary
is in particular of importance for applications of prophet
inequalities to posted pricing mechanisms. Indeed, one of
the main technical difficulties that Kleinberg and
Weinberg~\cite{kleinberg_2012_matroid} had to overcome
to apply their results to posted pricing mechanism,
was the fact that their results were 
only with respect to an online adversary.

%
% - Be clear about type of adversary.
%
% - Say that Kleinberg and Weinberg's algorithm even
%   works against a weight-adaptive adversary.
%   - Mention that ours works against a stronger
%     adversary, which is important for applications
%     in BMUMD.
% 
%
%

\subsubsection*{Oblivious posted pricing mechanisms\ifbool{SODACamera}{.}{}}

We start by introducing the Bayesian single-parameter
mechanism design setting (short BSMD), largely
following~\cite{chawla_2010_multi-parameter}.
There is a single seller providing a set
$N$ of services, and for each service
$e\in N$ there is one agent interested in $e$,
whose valuation is drawn from a nonnegative random
variable $Z_e$. The $Z_e$ are independent and
have known distributions. Furthermore, there is
a down-closed family $\mathcal{F}\subseteq 2^N$
representing feasibility constraints faced by
the seller, \ie, the seller can provide
any set of services $S\in \mathcal{F}$.
The setting is called \emph{single-parameter} because
every agent is interested in precisely one service.
The goal in this setting is to find truthful mechanisms maximizing
the expected revenue.

From a theoretical point of view, this setting
is well understood and optimally solved by Myersons's
mechanism~\cite{myerson_1981_optimal}.
Unfortunately, the resulting mechanism is impractical,
and thus, rarely employed. Furthermore, it does not
extend to multi-parameter settings where an agent
may, for example, be interested in buying one out of
several items, a setting known as \emph{Bayesian
multi-parameter unit-demand mechanism design}
(BMUMD).
Therefore, Chawla et
al.~\cite{chawla_2010_multi-parameter} suggested 
considerably simpler and more robust alternatives
%,
%in particular \emph{constrained order-oblivious
%posted price mechanisms} (COPM)---and more generally,
%sequential posted price mechanism---
having many advantages
while maintaining an almost optimal
performance~\cite{chawla_2010_multi-parameter,%
yan_2010_mechanism,%
kleinberg_2012_matroid}.
The idea is to offer to the agents sequentially
take-it-or-leave-it prices as follows.
Agents are considered one by one, in an
order chosen by the algorithm. Whenever an agent $e\in N$
is considered, the algorithm either makes no offer to
$e$---and thus $e$ does not get served---or,
if $e$ can be added to the elements selected so
far without violating feasibility, an offer
$p_e\in \mathbb{R}_{\geq 0}$ is made to $e$. Agent
$e$ will then accept the offer if $Z_e\geq p_e$ and
decline if $Z_e < p_e$.%
%\footnote{Selfishness of
%the agents only implies acceptance when $Z_e > p_e$
%and rejection when $Z_e < p_e$. The case $Z_e=p_e$
%is not of special importance here, and can be
%avoided by
%offering prices $p_e$ such that $\Pr[Z_e=p_e]=0$.}

This type of mechanism, with the additional freedom
that the algorithm can choose the order in which
to consider the agents, is called a
\emph{sequential posted price mechanism}.
A natural stronger version of sequential
posted price mechanisms, called
\emph{constrained oblivious posted price mechanisms}
(COPM), suggested by Chawla et
al.~\cite{chawla_2010_multi-parameter}, allows for
dealing with the multi-parameter setting BMUMD, and
has many further interesting properties.
Formally, a COPM is defined by a tuple
$(p\in \mathbb{R}_{\geq 0}^E, \mathcal{F}')$,
where $p$ are the take-it-or-leave-it prices,
and $\mathcal{F}'\subseteq \mathcal{F}$.
A COPM defined by $(p,\mathcal{F}')$ works as follows.
Consider the moment when a new agent $e$ arrives and
let $S$ be the set of agents served so far. If
$S \cup \{e\} \not\in \mathcal{F}'$, then $e$ is skipped;
otherwise, $e$ is offered the price $p_e$.
In short the COPM maintains a feasible set in the
more restricted family $\mathcal{F}'$, and greedily selects
any agent $e$ that does not destroy feasibility
in $\mathcal{F}'$ and has a valuation of at least $p_e$.
Furthermore, the order of the agents is chosen by
an adversary at the beginning of the
procedure, knowing all valuations $z_e$, the prices
$p_e$ and the family $\mathcal{F}'$.
A COPM can also be randomized, in which case the
tuple $(p, \mathcal{F}')$ is chosen at random
at the beginning of the algorithm.

So far, COPMs with an $O(1)$ gap with respect to the
optimal mechanism were only known for very restricted
types of matroids, and the intersection of two
partition matroids~\cite{chawla_2010_multi-parameter}.
For general matroids, the best previously
known COPM was a non-efficient
procedure with an optimality gap of $O(\log(\rank))$,
where $\rank$ is the size of a largest feasible
set of the matroid~\cite{chawla_2010_multi-parameter}.
In particular, the existence of an COPM for general
matroids with constant optimality gap remained
an open problem.

Exploiting a link between greedy OCRSs and COPMs,
we resolve the open question
about $O(1)$-optimal COPMs for matroids raised
in~\cite{chawla_2010_multi-parameter}, and show that
even much more general constraint families admit
$O(1)$-optimal COPMs.

\begin{theorem}\label{thm:copm}
Let $\mathcal{F}\subseteq 2^N$ be a down-closed family
and $P$ be a relaxation of $\mathcal{F}$.
If there is a $c$-selectable greedy OCRS for $P$,
then there is a COPM for $\mathcal{F}$ that,
even against an almighty adversary,
is at most a factor of $c$ worse than the
optimal truthful mechanism.
\end{theorem}

Using the reduction from the multi-parameter setting
to the single-parameter setting presented
in~\cite{chawla_2010_multi-parameter} we obtain results
for BMUMD under very general feasibility constraints.
\begin{corollary}
Let $\mathcal{F}$ be the intersection of a constant
number of matroid, knapsack, and matching constraints.
Then there is a posted price mechanism for BMUMD
on $\mathcal{F}$ 
whose optimality gap with respect to
the optimal truthful mechanism is at most a constant.
\end{corollary}

\ifbool{shortVersion}{
Moreover, as we discuss in the long version of the paper,
}
{Moreover, as we highlight in Section~\ref{sec:app},}
the mechanisms obtained through our greedy OCRSs can
be implemented efficiently under mild assumptions.

\subsubsection*{Stochastic probing\ifbool{SODACamera}{.}{}}

Recently, Gupta and Nagarajan~\cite{gupta_2013_stochastic}
introduced the following stochastic probing model.
Given is a finite ground set $N$ and each element
$e\in N$ is \emph{active} with a given
probability $p_e\in [0,1]$, independently of the
other elements.
Furthermore, there is a weight function
$w:N\rightarrow \mathbb{Z}$, and two down-closed constraint
families $\mathcal{F}_{in}, \mathcal{F}_{out}\subseteq 2^N$
on $N$, which are called the \emph{inner} and
\emph{outer constraints}, respectively.
The goal is to select a subset of active elements of
high weight according to the following rules.
An algorithm must first \emph{probe} an element $e\in N$
to select it. If a probed element $e$ is active,
then $e$ is selected, otherwise it is not.
The algorithm can choose the order in which elements
are probed.
The set $Q$ of all probed elements must satisfy
$Q\in \mathcal{F}_{out}$ and the set $S\subseteq Q$
of all selected elements must satisfy
$S\in \mathcal{F}_{in}$. Hence, at any step of the
algorithm an element can only be probed if adding it
to the currently probed elements does not violate
$\mathcal{F}_{out}$, and adding it to the elements
selected so far does not violate $\mathcal{F}_{in}$.
Gupta and Nagarajan~\cite{gupta_2013_stochastic} showed
that this model captures numerous applications.
Furthermore, they show how so-called
\emph{ordered CRSs} can be used to get
approximately optimal approximation factors for
many constraint families including the intersection
of a constant number of matroids.
However, due to their use of ordered CRSs, the
presented algorithms crucially rely on the fact that
the order in which elements are probed can be chosen
freely. 

Replacing their use of ordered CRSs with OCRSs
we can drop this requirement.

\begin{theorem}\label{thm:ocrsProbing}
Let $\mathcal{F}_{in}, \mathcal{F}_{out} \subseteq 2^N$
be two down-closed families.
If there are a $(b, c_{in})$-selectable greedy OCRS $\pi_{in}$ for a relaxation $P_{in}$ of $\cF_{in}$ and a $(b, c_{out})$-selectable greedy OCRS $\pi_{out}$ for a relaxation $P_{out}$ of $\cF_{out}$, then there
is a $(b \cdot c_{in} \cdot c_{out})$-approximation for the weighted stochastic
probing problem where the order in which elements
can be probed is chosen by an almighty adversary and the inner and outer
constraints are given by
$\mathcal{F}_{in}$ and $\mathcal{F}_{out}$,
respectively.\footnote{Similar to~\cite{gupta_2013_stochastic} one
can strengthen the theorem, and only assume
an offline CRS for $\mathcal{F}_{out}$ and
an OCRS for $\mathcal{F}_{in}$.
The fact that an offline CRS suffices
for $\mathcal{F}_{out}$ can sometimes be used to get
better approximation factors.
For simplicity of presentation,
we do not go into these details here,
and also in Theorem~\ref{thm:ocrsProbDeadline}
and~\ref{thm:ocrsProbSubm}.}
Moreover, if $\pi_{in}$ and $\pi_{out}$ are efficient and there are separation oracles for $P_{in}$ and $P_{out}$, then the above algorithm has a polynomial time complexity.%
\end{theorem}

It turns out that the extension to arbitrary probing orders
resolves an open question of~\cite{gupta_2013_stochastic}
about stochastic probing with deadline.
In a probing problem with deadlines there is a
deadline $d_e\in \mathbb{Z}_{\geq 1}$ for each element
$e\in N$, indicating that $e$ can
only be probed as one of the first $d_e$ elements that
get probed.
Using a clever technique, Gupta and
Nagarajan~\cite{gupta_2013_stochastic} presented an
$O(1)$-approximation for this problem setting for
the unweighted case, \ie, $w$ is the all-ones vector,
when $\mathcal{F}_{in}, \mathcal{F}_{out}$ are
$k$-systems\footnote{A $k$-system $\mathcal{F}\subseteq 2^N$
is a down-closed family such that, for any $S\subseteq N$,
the ratio of the sizes of any two maximal sets
of $\mathcal{F}$ that are contained
in $S$ is at most $k$.
In particular, $k$-systems generalize the intersection of
$k$ matroids.},
for $k=O(1)$.
They left it as an open question how to approach
the weighted version of stochastic probing with
deadlines. Using OCRSs we can leverage
Theorem~\ref{thm:ocrsProbing} to consider elements
in increasing order of their deadlines, which allows
for addressing this open question.
%
%Furthermore, our results on using OCRSs for
%submodular maximization, directly imply an
%extension of stochastic probing with deadlines
%to submodular objectives.

\begin{theorem}\label{thm:ocrsProbDeadline}
Let $\mathcal{F}_{in}, \mathcal{F}_{out} \subseteq 2^N$
be two down-closed families.
If there are a $(b, c_{in})$-selectable greedy OCRS $\pi_{in}$ for a relaxation $P_{in}$ of $\cF_{in}$ and a $(b, c_{out})$-selectable greedy OCRS $\pi_{out}$ for a relaxation $P_{out}$ of $\cF_{out}$, then there
is a $(b(1 - b) \cdot c_{in} \cdot c_{out})$-approximation for the weighted stochastic
probing with deadlines problem where the inner and outer
constraints are given by
$\mathcal{F}_{in}$ and $\mathcal{F}_{out}$,
respectively.
Moreover, if $\pi_{in}$ and $\pi_{out}$ are efficient and there are separation oracles for $P_{in}$ and $P_{out}$, then the above algorithm has a polynomial time complexity.
\end{theorem}

We highlight that the stochastic probing problem
with monotone submodular objectives but without
deadlines---in short,
\emph{submodular stochastic probing}---was 
%Moreover, our results on OCRSs for submodular
%maximization readily imply that our algorithms
%for stochastic probing can also deal with
%submodular objective functions.
%
%This submodular stochastic probing problem
considered by Adamczyk et
al.~\cite{adamczyk_2014_submodular}% (and some further
%results on non-monotone submodular
%maximization have very recently been announced
%by Adamczyk~\cite{adamczyk_2015_non-negative})
, who presented
for this setting a
$(1-1/e)/(k_{in} + k_{out} + 1)$-approximation when
$\mathcal{F}_{in}$ and $\mathcal{F}_{out}$ are the
intersection of $k_{in}$ and $k_{out}$ matroids,
respectively.
Using our techniques we obtain $O(1)$-approximations
for considerably more general settings of submodular
stochastic probing. More precisely, 
we can handle a very broad set of constraints,
with a probing order chosen by an
almighty adversary (instead of being choosable
by the algorithm).

%and the submodular function we maximize may be non-monotone.

\begin{theorem}\label{thm:ocrsProbSubm}
Let $\mathcal{F}_{in}, \mathcal{F}_{out} \subseteq 2^N$
be two down-closed families.
If there are a $(b, c_{in})$-selectable greedy OCRS $\pi_{in}$ for a relaxation $P_{in}$ of $\cF_{in}$ and a $(b, c_{out})$-selectable greedy OCRS $\pi_{out}$ for a relaxation $P_{out}$ of $\cF_{out}$, then there
is a $((1 - e^{-b} - o(1)) \cdot c_{in} \cdot c_{out})$-approximation for the submodular stochastic
probing problem where the order in which elements
can be probed is chosen by an almighty adversary and the inner and outer
constraints are given by
$\mathcal{F}_{in}$ and $\mathcal{F}_{out}$,
respectively. Moreover, if $\pi_{in}$ and $\pi_{out}$ are efficient and there are separation oracles for $P_{in}$ and $P_{out}$, then the above algorithm has a polynomial time complexity.
\end{theorem} 

We remark that the same idea used to derive Theorem~\ref{thm:ocrsProbDeadline} from Theorem~\ref{thm:ocrsProbing} can also be used to derive from Theorem~\ref{thm:ocrsProbSubm} a result for the submodular stochastic
probing with deadlines problem.

%% file: contres.tex
\ifbool{shortVersion}{
\section{Online Contention Resolution Scheme for Matroids}

In this section we present an OCRS for matroids, which
leads to the result claimed in Theorem~\ref{thm:direct_OCRS}
about matroids. Due to space constraints, OCRSs for further
constraints and missing proofs are deferred to the long
version of the paper.

}{
\section{Constructing \ifbool{SODACamera}{OCRSs}{Online Contention Resolution Schemes}}

In this section we prove the existence (or non-existence) of OCRSs for various polytopes. Sections~\ref{ssc:matroids}, \ref{ssc:matching} and~\ref{ssc:knapsack} study OCRSs for matroid polytopes, matching polytopes and the natural relaxation of knapsack constraints, respectively. The results proved in these sections prove together Theorem~\ref{thm:direct_OCRS}. Theorem~\ref{thm:combineOCRSs}, which shows that greedy OCRSs for different polytopes can be combined to create greedy OCRSs for the intersection of these polytopes, is proved in Section~\ref{ssc:combine}.

\subsection{OCRS for matroids\ifbool{SODACamera}{.}{}} \label{ssc:matroids}

In this section we give a greedy OCRS for matroid polytopes. }%end ifbool
For standard matroidal
concepts such as $\spn$, $\rank$, contraction and restriction, we refer the
reader to
\ifbool{shortVersion}{Appendix~A in the long version of the paper}{Appendix~\ref{app:matroid_definitions}}.
Also recall that, for a given matroid $M=(N,
\mathcal{F})$, the matroid polytope $P_\mathcal{F}$ is defined by $\{x\in
  \mathbb{R}_{\geq 0}^N \mid \forall S \subseteq N \, \sum_{e\in S} x_e \leq \rank(S)\}$.
The main result of this section can be stated as follows.

\begin{theorem}
  Let $b\in [0,1]$. There exists a $(b,1-b)$-selectable deterministic greedy OCRS
  for any matroid polytope $P_\mathcal{F} \subseteq [0,1]^N$ on ground set $N$.
  \label{thm:matroidocrs}
\end{theorem}
Combining Theorem~\ref{thm:combineOCRSs} with
Theorem~\ref{thm:matroidocrs}, we also get a greedy OCRS
for the intersection of $k$ matroids.

\begin{corollary}\label{cor:ocrsSeveralMatroids}
  Let $b\in [0,1]$, and let
  $P_1, \dots, P_k \subseteq [0,1]^N$ be $k$ matroid
  polytopes over a common ground set $N$. 
  Then there exists a $(b,(1-b)^k)$-selectable deterministic greedy
  OCRS for $P=\cap_{i=1}^k P_i$. 
\end{corollary}

Note that the above corollary implies (by Observation~\ref{obs:pair_to_single_selectability}) that there is a
$b(1-b)^k$-selectable greedy OCRS for $P$.
Moreover, 
choosing $b=\frac{1}{1+k}$ in
Corollary~\ref{cor:ocrsSeveralMatroids}, we obtain
the following.
\begin{corollary}
  Let $P_1,\dots, P_k\subseteq [0,1]^N$ be
  matroid polytopes over a common ground set $N$,
  let $P=\cap_{i=1}^k P_i$,
  and let $c=\frac{1}{k+1} (1-\frac{1}{1+k})^k\geq
  \frac{1}{e(k+1)}$.
  Then there exists a $c$-selectable greedy
  OCRS for $P$. 
\end{corollary}

The rest of this section is devoted to the proof of Theorem~\ref{thm:matroidocrs}.
Consider a matroid $M = (N, \mathcal{F})$ and let $P_\mathcal{F}$ be the associated polytope.
Our greedy OCRS is based on using $x\in P_\mathcal{F}$ to find a  \emph{chain decomposition} of the elements
\begin{align*}
  \varnothing = N_\ell \subsetneq N_{\ell-1} \subsetneq \dots \subsetneq N_1 \subsetneq N_0 = N
	\enspace.
\end{align*}
It then accepts an active element $e\in N_{i} \setminus N_{i+1}$ if $e$ together with the
already accepted elements in $N_i \setminus N_{i+1}$ forms an independent set in
the matroid $(M/N_{i+1})|_{N_i}$, \ie, the matroid obtained from $M$ by
contracting $N_{i+1}$ and then restricting to $N_i$. To see that this OCRS is a greedy OCRS, note that the above algorithm is
equivalent to defining the family $\mathcal{F}_x = \{I\subseteq N: \forall i\,
  I\cap (N_i \setminus N_{i+1})\mbox{ is independent in
  } (M/N_{i+1})|_{N_i}\}$. The family $\mathcal{F}_x$ is clearly  a down-closed
  family of sets (since each $(M/N_{i+1})|_{N_i}$ is a matroid and its
  independent sets are, thus, down-closed).  Moreover, $\mathcal{F}_x$  is
  a subset of feasible sets because (see, \eg, Theorem~$5.1$ in~\cite{soto_2013_matroid})  if  $I_i$ is an independent set of
  $(M/N_{i+1})|_{N_i}$ for every $i$, then the set $\cup_{i} I_i$ is independent
  in $M$.
  Even though we do not need this fact, we highlight that
  $\mathcal{F}_x$ itself describes a family of independent sets
  of a matroid. This follows from the fact that
  $\mathcal{F}_x$ is the family of all (disjoint) unions of
  independent sets in the matroids $(M/N_{i+1})|_{N_i}$
  for $i\in \{0,\dots, \ell-1\}$, implying that 
  $\mathcal{F}_x$ are the independent sets
  of the union matroid obtained by taking the union of
  all matroids  $(M/N_{i+1})|_{N_i}$ for
  $i\in \{0,\dots, \ell-1\}$.

  Having shown that the above algorithm is a greedy OCRS for any chain
  decomposition of the elements, we turn our attention to the task of defining
  a chain that maximizes the selectability of our greedy OCRS. For a fixed
  chain, we have that the selectability of an element $e\in N_{i}\setminus N_{i+1}$ is
  \ifbool{SODACamera}{
  \begin{align*}
    \Pr[ I \cup & \{e\}  \in \mathcal{F}_x \;\;\;
    \forall \;I\subseteq R(x),
    I\in \mathcal{F}_x] \\
    & = \Pr[e\not \in \spn_i\left((R(x) \cap (N_i \setminus N_{i+1})) \setminus \{e\}\right)]
		\enspace,
  \end{align*}
  }
  {
  \begin{align*}
    \Pr[ I \cup \{e\} \in \mathcal{F}_x \;\;\;
    \forall \;I\subseteq R(x),
    I\in \mathcal{F}_x]
    & = \Pr[e\not \in \spn_i\left((R(x) \cap (N_i \setminus N_{i+1})) \setminus \{e\}\right)]
		\enspace,
  \end{align*}
  }
  where $\spn_i(\cdot)$ denotes the span function of matroid $(M/N_{i+1})|_{N_i}$.
Our objective is therefore to construct a chain decomposition
that
maximizes $\Pr[e\not
\in \spn_i\left((R(x) \cap (N_i \setminus N_{i+1})) \setminus \{e\}\right)]$
or equivalently minimizes
\begin{align}
%  \Pr[e\not \in \spn_i\left((R(x) \cap (N_i \setminus E_{i+1})) \setminus \{e\}\right)]
%  \geq 1-b \Leftrightarrow 
  \Pr[e \in \spn_i\left((R(x) \cap (N_i \setminus N_{i+1})) \setminus \{e\}\right)]%\leq b.
	\enspace. 
  \label{eq:spanprob}
\end{align}
%for all elements which gives our $(b,1-b)$-selectable greedy OCRS as stated in Theorem~\ref{thm:matroidocrs}.

We now describe an iterative procedure for constructing  a chain decomposition so
that~\eqref{eq:spanprob} is at most $b$ for each element. Initially,
the chain only consists of the ground set $N_0 = N$. We need to refine the
chain if there exists an element $e\in N$ such that 
$\Pr[e \in \spn(R(x) \setminus \{e\})] > b$. We do that in a ``minimal'' way as follows:
\begin{itemize}
\setlength\itemsep{0em}
  \item Let $S = \varnothing$.
   \item While there exists $e\in N_{0} \setminus S$ such that $\Pr[ e\in \spn\left( (R(x)
     \cup S) \setminus \{e\} \right)] > b$, add $e$ to $S$.
\end{itemize}

We then let $N_1 = S$. Note that if $N_1$ is a strict subset of $N_0$ then we
have made progress and we repeat the above procedure (on the matroid induced by $N_1$) to obtain $N_2\subsetneq
N_1$, and so on, until we obtain $N_\ell = \varnothing$.
Thus, if we assume that the procedure terminates, then, by construction, the chain satisfies
\ifbool{SODACamera}
{
\begin{align*}
  \Pr[\spn_{i}\left( (R(x) \cap (N_{i} \setminus N_{i+1}) ) \setminus \{e\} \right) ] \leq b
\end{align*}
 for all  $i$ and  $e\in N_{i} \setminus N_{i+1}$.
}
 {
\begin{align*}
  \Pr[\spn_{i}\left( (R(x) \cap (N_{i} \setminus N_{i+1}) ) \setminus \{e\} \right) ] \leq b \qquad \mbox{for all } i \mbox{ and } e\in N_{i} \setminus N_{i+1}
	\enspace.
\end{align*}
}

\ifbool{shortVersion}{
The long version of the paper shows why the above chain
construction terminates which, together with the above discussion,
implies a $(b,1-b)$-selectable greedy OCRS for matroids,
thus proving Theorem~\ref{thm:matroidocrs}.
}{
The rest of this section is devoted to show that the chain construction always
terminates. As described above, this implies
a $(b,1-b)$-selectable greedy OCRS, and thus proves Theorem~\ref{thm:matroidocrs}.

\subsubsection{Proof of termination of chain construction\ifbool{SODACamera}{.}{}}

%%% BEGIN HIDE %%%%
\iffalse
We start with a simple proposition.

\begin{proposition}
  \label{prop:gendist}
  For any point $x\in b \cdot P_\mathcal{F}$ and any distribution $\mu$ over $2^N$,
  \begin{align*}
    \sum_{e\in N} x_e \Pr_{A\sim \mu} [e \in \spn(A)]  \leq b\cdot \E_{A \sim \mu}[|A|].
  \end{align*}
\end{proposition}
\begin{proof}
  The proposition follows from basic calculations.
  
  First, by definition,  $\sum_{e\in N} x_e \Pr_{A\sim \mu} [e \in \spn(A)]
  = \E_{A \sim \mu}[x(\spn(A))]$.
  Second, using that $x \in b\cdot P_\mathcal{F}$,
  \begin{align*}
    \E_{A \sim \mu}[x(\spn(A))]  \leq b\cdot \E_{A \sim \mu}[\rank(\spn(A))],
  \end{align*}
  which equals $b\cdot \E_{A \sim \mu}[\rank(A)] \leq b\cdot \E_{A \sim \mu}[|A|]$.

\end{proof}

\fi
%%% END HIDE %%%
To prove that the chain decomposition terminates, it is sufficient to show that
$N_1 = S \subsetneq N_0 = N$ for any (non-empty) matroid  as the chain decomposition
then recurses on the matroid induced by $N_1$ (so the same argument implies that
$N_2 \subsetneq N_1$ assuming $N_1 \neq \varnothing$, and so on). Notice that the definition of $S$ implies that $S$ can only increase as coordinates of $x$ are increased. Hence, it is safe to assume that $x \in b\cdot P_{\cB}$,
where $P_{\cB} = \{x\in P_\mathcal{F} \mid x(N) = \rank(N)\}$ is the base
polytope of the matroid $M$,
which is the set of all maximal vectors in $P_{\mathcal{F}}$.

For proofs in the rest of the section it will be convenient to have the following equivalent view of the construction of $S = N_1$ in the refinement procedure.
\begin{itemize}
\setlength\itemsep{0em}
  \item Let $S_0 = \varnothing$, and let $S_1 = \{e \in N\mid \Pr[e \in \spn\left( R(x) \setminus \{e\}
    \right)] > b\}$ be the set of elements that are ``likely'' to be spanned.
  \item Assuming we have defined $S_0, S_1, \ldots, S_{i-1}$,  let
    \ifbool{SODACamera}
    {$S_i$ be
    \begin{align*}
      \{e \in N\mid \Pr[e \in \spn\left( (R(x) \cup S_{i-1}) \setminus \{e\} \right)] > b \}
			,
    \end{align*}
    }
    {
    \begin{align*}
      S_i  & =  \{e \in N\mid \Pr[e \in \spn\left( (R(x) \cup S_{i-1}) \setminus \{e\} \right)] > b \}
			\enspace,
    \end{align*}
  }
    that is, $S_i$ contains those elements that are likely to be spanned assuming
    that the elements of $S_{i-1}$ are contracted (or equivalently appear with
    probability $1$).
\end{itemize}
Notice that $S_{i-1} \subseteq S_{i}$ for every $i \geq 1$, and $S_i = S_{i - 1}$ implies $S_i = S_j$ for every $j > i$. Thus, we must have $N_1 = S = S_{|N|}$. The key technical part of the termination analysis is the following lemma.
\begin{lemma} \label{lem:technical}
  It always holds that:
  \ifbool{SODACamera}{
  \begin{align*}
    \sum_{e\in N} x_e \Pr[&e\in \spn\left( R(x) \cup S \right)] \\
    & \leq b\cdot\left(  x(N)  + (1-b) \rank(S) \right)
		\enspace.
  \end{align*}
  }
  {
  \begin{align*}
    \sum_{e\in N} x_e \Pr[e\in \spn\left( R(x) \cup S \right)] \leq b\cdot\left(  x(N)  + (1-b) \rank(S) \right)
		\enspace.
  \end{align*}
  }
  Moreover, the inequality is strict if $S\neq \varnothing$. 
\end{lemma}
Before proving the lemma, let us see that it implies what we want:
\begin{corollary}
If $N\neq \varnothing$  then $N_1 = S \subsetneq N$.
\end{corollary}
\begin{proof}
The corollary is clearly true if $S= \varnothing$. Otherwise, by Lemma~\ref{lem:technical}, 
\ifbool{SODACamera}
{
\begin{align*}
  x(S) &\leq{}  x(\spn\left( R(x) \cup S \right)) \\
  &= \sum_{e\in N} x_e \Pr[e\in \spn\left( R(x) \cup S \right)]  \\ & <{}   b (x(N)  + (1-b)
      \rank(S)) \leq  b \rank(N)
			\enspace.
\end{align*}
}
{
\begin{align*}
  x(S) \leq{} & x(\spn\left( R(x) \cup S \right)) = \sum_{e\in N} x_e \Pr[e\in \spn\left( R(x) \cup S \right)]  \\<{} &  b (x(N)  + (1-b)
      \rank(S)) \leq  b \rank(N)
			\enspace.
\end{align*}
}
As $x(S) < b\cdot\rank(N)$ and $x(N) = b\cdot\rank(N)$ by our assumption that $x \in bP_\cB$, we get $x(S) < x(N)$, which implies $N\setminus S \neq \varnothing$.
\end{proof}

Let us now continue with the proof of the our technical lemma.
\begin{proof}[Proof of Lemma~\ref{lem:technical}]
  Let $S' = \{e_1, e_2, \dots, e_k\}$ be a basis of the matroid $M|_{S}$
   obtained by first greedily selecting elements from $S_0$, then greedily adding elements from $S_1$, and so on.
  Consider the distribution $\mu$ over $2^N$ defined according to the following
  sampling procedure: 
  \begin{compactenum}
    \item Let $A$ be a random set originally distributed like $R(x)$.
    \item For $j=1, 2, \dots, k$, if $e_j \not \in \spn(A)$, add $e_j$ to $A$.
    \item Output $A$.
  \end{compactenum}
  Observe that this sampling procedure guarantees that the distributions of $\spn(S \cup R(x))$,  $\spn(S' \cup R(x))$ and  $\spn(A)$ are identical. Therefore,
  \begin{align*}
    \Pr[e\in \spn\left( R(x) \cup S \right)]  = {\textstyle \Pr}_{A\sim \mu} [e \in \spn(A)]
		\enspace.
  \end{align*}
  Now simple calculations yield
%  This together with Proposition~\ref{prop:gendist} yields  
  \ifbool{SODACamera}
  {
    \begin{align*}
      {\textstyle \sum}_{e\in N}\mspace{-9mu}&\mspace{9mu} x_e {\textstyle \Pr}_{A \sim \mu}[e \in \spn(A)]  
       = \E_{A \sim \mu} [ x(\spn(A))]  \\
      & \leq b\cdot \E_{A \sim \mu} [\rank(\spn(A))]  \mbox{ \quad  (using $x\in b \cdot P_\cI$)}  \\
      & = b\cdot \E_{A \sim \mu} [\rank(A)]\leq b\cdot \E_{A \sim \mu} [|A|]
			\enspace.
    \end{align*}
  }
  {
    \begin{align*}
      {\textstyle \sum}_{e\in N} x_e {\textstyle \Pr}_{A \sim \mu}[e \in \spn(A)] 
      & = \E_{A \sim \mu} [ x(\spn(A))]  \\
      & \leq b\cdot \E_{A \sim \mu} [\rank(\spn(A))] & \mbox{(using $x\in b \cdot P_\cI$)}  \\
      & = b\cdot \E_{A \sim \mu} [\rank(A)]\leq b\cdot \E_{A \sim \mu} [|A|]
			\enspace.
    \end{align*}
  }
  We complete the proof by showing that $\E_{A \sim \mu}[|A|] \leq x(N) + (1-b) \rank(S)$. To see this, note that
  \ifbool{SODACamera}{$\E_{A \sim \mu}[|A|]$ equals
  \begin{align*}
    x(N&) + \sum_{j=1}^k \Pr[e_j \not \in \spn(R(x) \cup \{e_1, \dots, e_{j-1}\})]  \\
     &\leq x(N) + (1-b) k = x(N) + (1-b) \cdot \rank(S)
		\enspace,
  \end{align*}
  }
  {
  \begin{align*}
    \E_{A \sim \mu}[|A|]  &= x(N) + \sum_{j=1}^k \Pr[e_j \not \in \spn(R(x) \cup \{e_1, \dots, e_{j-1}\})]  \\
    & \leq x(N) + (1-b) k = x(N) + (1-b) \cdot \rank(S)
		\enspace,
  \end{align*}
}
  where the inequality follows from the following argument. If we let $i$ be the smallest index
  so that $e_j \in S_i$, then by the construction of $S'$ we have $\Pr[e_j \in
  \spn(R(x) \cup \{e_1, \dots, e_{j-1}\})] \geq \Pr[e_j \in \spn(R(x) \cup
  S_{i-1})] > b$. Finally, we remark that the strict  inequality $\Pr[e_j \in \spn(R(x) \cup
  S_{i-1})] > b$  implies that the inequality in the statement is strict
  if $S \neq \varnothing$.
\end{proof}

\subsubsection*{Efficient implementation\ifbool{SODACamera}{.}{}}

The only step that is not constructive in
the description of our OCRS for matroids is
the computation of probabilities of the type
$\Pr[e\in \spn((R(x)\cup S_{i-1})\setminus \{e\})]$.
One can easily get around this issue by using
good estimates through Monte-Carlo sampling,
leading to the following.

\begin{lemma}
For any $\epsilon > 0$ and $\alpha > 0$, there
is a randomized construction of a chain
$\emptyset = N_\ell \subsetneq N_{\ell-1}
\subsetneq \ldots \subsetneq N_1 \subsetneq N_0 = N$
such that with probability at least
$1-|N|^{-\alpha}$
the greedy OCRS defined by the chain
is $(b,1-b-\epsilon)$-selectable.
Furthermore, the time needed for this construction is
$O(\alpha \cdot \frac{1}{\epsilon^2} \cdot \poly(|N|)
 \cdot T)$, where $T$ is the
time for a single call to the independence
oracle of the matroid, and we assume that we
can sample a Bernoulli random variable in $O(1)$ time.
\end{lemma}
\ifbool{SODACamera}
{The proof of the above lemma is deferred to the full version of this paper~\cite{feldman_2015_online}.
}
{
\begin{proof}
For simplicity we define
$p_{e,S} = \Pr[e\in \spn((R(x)\cup S)\setminus \{e\})]$
for $e\in N$ and $S\subseteq N$.
As before, we focus on the construction of $N_1$; the
algorithm is then applied recursively.
We recall that in the above-mentioned construction
we set $N_1 = S$, where $S$ was constructed from
$S=\varnothing$ by adding elements 
$e\in N$ satisfying $p_{e,S} > b$.
To perform this step constructively we will,
whenever we need a probability $p_{e,S}$, use
an estimate $\hat{p}_{e,S}$ obtained
through Monte-Carlo sampling.
By standard results on Monte-Carlo sampling
(see, \eg, \cite{dagum_2000_optimal})
it suffices to use
$O(\alpha \cdot \frac{1}{\epsilon^2} \cdot \log|N|)$
many samples to obtain a value $\hat{p}_{e,S}$ such
that
\begin{align*}
\Pr\left[\hat{p}_{e,S} \in [p_{e,S}-\epsilon,
 p_{e,S}]\right] \geq 1-|N|^{-3-\alpha}
\enspace.
\end{align*}
Hence, to construct $S$, we start with $S=\varnothing$
and successively add elements $e\in N\setminus S$ with
$\hat{p}_{e,S} > b$.
There are at most $|N|$ elements we add to $S$, and
to add one element to $S$ we may have to check
the values $\hat{p}_{e,S}$ of all elements in
$N\setminus S$. Hence, to construct $S$ we use
at most $O(|N|^2)$ estimates of the type $\hat{p}_{e,S}$.
At the end of this procedure we set $N_1=S$ and
repeat. Since there are at most $O(|N|)$
sets in the final chain
$\emptyset=N_\ell \subsetneq \ldots \subsetneq N_0=N$,
the total number of estimates we need is bounded
by $O(|N|^3)$, which implies the claimed running time
of our algorithm. Moreover, with probability at least
$1-|N|^{-\alpha}$, all our estimates $\hat{p}_{e,S}$ satisfy
$\hat{p}_{e,S} \in [p_{e,S}-\epsilon, p_{e,S}]$.
To prove the lemma, we assume from now on that all our
estimates fulfil this property and show that this
implies that the OCRS we obtain
is $(b,1-b-\epsilon)$-selectable.

Notice that during our construction of $S$
we only add elements $e$ to $S$ satisfying
$b<\hat{p}_{e,S}\leq p_{e,S}$. Hence, the elements
we add would also have been added to $S$ in the
construction that uses the true probabilities $p_{e,S}$.
Therefore, for the same reasons showing that $S\subsetneq N$
when using the true probabilities $p_{e,S}$, we also have
$S\subsetneq N$.
It remains to observe that at the end of the construction
of $S$, \ie, when we set $N_1=S$, the probability
$1-p_{e,S}=\Pr[e\not\in \spn((S\cup R(x))\setminus \{e\})]$
of an element $e\in N\setminus S$ being selectable
is at least $1-b-\epsilon$. This indeed holds since
\[
1-p_{e,S} \geq 1 - \hat{p}_{e,S} - \epsilon \geq 1-b-\epsilon \enspace.
\ifbool{SODACamera}{}{\qedhere}
\]
\end{proof}
}

\subsection{OCRSs for matchings in general graphs\ifbool{SODACamera}{.}{}} \label{ssc:matching}

In this section we describe an OCRS that works for a relaxation $P_G$ of matching in a general graph $G=(V,E)$. Specifically, the relaxation polytope $P_G$ is defined as:
\[
	\begin{array}{lll}
		\sum_{g \in \delta(u)} x_g & \leq 1 & \forall\; u \in V \\
		x_g & \geq 0 & \forall\; g \in E \enspace,
	\end{array}
\]
where $\delta(u)$ is the set of edges incident to the node $u$. Observe that this relaxation is weaker than the matching polytope, hence, our results hold also for the matching polytope. We would like to stress that the ground set in this section is the set of edges, and thus, unlike in the rest of the paper, we denote it by $E$. Some ideas from the proof of Theorem~\ref{thm:matching_OCRS} can be traced to offline CRSs given by~\cite{feldman_2011_unified}.\footnote{The details of these CRSs are omitted in~\cite{feldman_2011_unified}, but can be found in~\cite{feldman_2011_maximization}.}

\begin{theorem} \label{thm:matching_OCRS}
  For every $b \in [0,1]$, there exists a $(b, e^{-2b})$-selectable randomized greedy OCRS for the relaxation $P_G\subseteq [0,1]^E$ of matching in a graph $G=(V,E)$.
\end{theorem}
\begin{proof}
	Let $x \in bP_G$ be the input point to the OCRS, and let $A \sim R(x)$ be the set of active elements. Our OCRS begins by selecting a subset $K$ of potential edges, where every edge $g \in E$ belongs to $K$ with probability $(1 - e^{-x_g}) / x_g$, independently (observe that this probability in indeed always within the range $[0, 1]$). Whenever an edge $g$ reveals whether it is active, the OCRS selects it if $g \in A \cap K$ and the addition of $g$ to the set of already selected edges does not make this set an illegal matching. Observe that for any fixed choice of $K$ this OCRS is a deterministic greedy OCRS, and thus, for a random $K$ it is a randomized greedy OCRS.

	Next, let us show that our OCRS is $(b, e^{-2b})$-selectable. Consider an arbitrary edge $g' = uv \in E$. We need to prove that with probability at least $e^{-2b}$ the edge $g'$ is in $K$, and can be added to any matching which is a subset of $A \cap K$. Formally, we need to prove:
	\[
		\Pr[g' \in K, A \cap K \cap (\delta(u) \cup \delta(v) \setminus \{g'\}) = \varnothing]
		\geq
		e^{-2b}
		\enspace.
	\]
	Clearly, every edge $g \in E$ belongs to $A \cap K$ with probability $x_g \cdot (1 - e^{-x_g}) / x_g = 1 - e^{-x_g}$. Since the membership of every edge in $K$ and $A$ is independent from the membership of other edges in these sets, we get:
  \ifbool{SODACamera}
  {
	\allowdisplaybreaks
	\begin{align*}
		\Pr[g' &\in K, A \cap K \cap (\delta(u) \cup \delta(v) \setminus \{g'\}) = \varnothing] \\
		&=
		\Pr[g' \in K] \cdot \prod_{g \in \delta(u) \cup \delta(v) \setminus \{g'\}} \mspace{-27mu} \Pr[g \not \in A \cap K]\\
		& ={} 
		\frac{(1 - e^{-x_{g'}})}{x_{g'}} \cdot \prod_{g \in \delta(u) \cup \delta(v) \setminus \{g'\}} \mspace{-27mu} e^{-x_g} \\
		& =
		\frac{(1 - e^{-x_{g'}})}{x_{g'}} \cdot e^{-\sum_{g \in \delta(u) \cup \delta(v) \setminus \{g'\}} x_g}\\
		& \geq{}
		\frac{(1 - e^{-x_{g'}})}{x_{g'}} \cdot e^{-2(b - x_{g'})} \\
		& =
		\frac{e^{x_{g'}}(e^{x_{g'}} - 1)}{x_{g'}} \cdot e^{-2b}
		\geq
		e^{-2b}
		\enspace,
	\end{align*}
  }
  {
	\begin{align*}
		&
		\Pr[g' \in K, A \cap K \cap (\delta(u) \cup \delta(v) \setminus \{g'\}) = \varnothing]
		=
		\Pr[g' \in K] \cdot \prod_{g \in \delta(u) \cup \delta(v) \setminus \{g'\}} \mspace{-27mu} \Pr[g \not \in A \cap K]\\
		={} &
		\frac{(1 - e^{-x_{g'}})}{x_{g'}} \cdot \prod_{g \in \delta(u) \cup \delta(v) \setminus \{g'\}} \mspace{-27mu} e^{-x_g}
		=
		\frac{(1 - e^{-x_{g'}})}{x_{g'}} \cdot e^{-\sum_{g \in \delta(u) \cup \delta(v) \setminus \{g'\}} x_g}\\
		\geq{} &
		\frac{(1 - e^{-x_{g'}})}{x_{g'}} \cdot e^{-2(b - x_{g'})}
		=
		\frac{e^{x_{g'}}(e^{x_{g'}} - 1)}{x_{g'}} \cdot e^{-2b}
		\geq
		e^{-2b}
		\enspace,
	\end{align*}
  }
	where the first inequality holds since the membership of $x$ in $bP_G$ guarantees that the total $x$-values of the edges in $\delta(u) \setminus \{g'\}$ (or $\delta(v) \setminus \{g'\}$) is at most $b - x_{g'}$.
\end{proof}

\noindent \textbf{Remark:} If one is interested in a \emph{deterministic} greedy OCRS, it is possible to set $K$ equal to $E$ deterministically in the greedy OCRS described by the last proof. A simple modification of the proof shows that the resulting deterministic greedy OCRS is $(b, (1-b)^2)$-selectable.

\subsection{OCRSs for knapsack constraint\ifbool{SODACamera}{.}{}} \label{ssc:knapsack}

In this section we consider the problem of defining an OCRS for a polytope $P\subseteq [0,1]^N$ defined by a single knapsack constraint. That is, each element $e\in N$ has an associated size $s_e \in [0,1]$ and $P$ is defined by
\[\begin{array}{lll}
  \sum_{e\in N} s_e x_e & \leq 1 & \\
  x_e & \in [0,1] & \mbox{for } e\in N
	\enspace. 
\end{array}\]

We begin with an interesting simple observation.

\begin{proposition}
  For every $n \geq 1$, there exists a knapsack constraint over a ground set of $n$ elements such that no deterministic greedy OCRS for the polytope defined by this constraint is $(b, c)$-selectable for any pair of $b \in [0, 1]$ and $c > (1-b)^{n-1}$.
\end{proposition}
\begin{proof}
  Consider the knapsack constraint defined by the ground set $N = \{1, 2, \dots, n\}$ and sizes $s_1 = s_2
  = \dots = s_{n-1} = 1/n$,  $s_n = 1$. Assume towards a contradiction that there exists a deterministic greedy OCRS that is $(b, c)$-selectable for some $b \in [0, 1]$ and $c > (1-b)^{n-1}$, and consider the family $\cF_x$ of feasible sets used by this OCRS for the possible input $x_1 = x_2 =\dots = x_{n-1} = b, x_n = b/n$.
	
  Since the OCRS is $(b, c)$-selectable for $c > 0$, each element $e \in N$ must be included in at least one set of the family $\cF_x$. Thus, by the down-monotonicity of $\cF_x$ it must contain the set $\{e\}$ for every $e \in N$. On the other hand, for every $e \in N \setminus \{n\}$ we have $\{e, n\} \not \in \cF$, and thus, also $\{e, n\} \not \in \cF_x$. Combining all these observations, we get:
  \ifbool{SODACamera}
  {
  \begin{align*}	
		\Pr[I \cup &\{n\} \in \mathcal{F}_x \;\;\; \forall \;I\subseteq R(x), I\in \mathcal{F}_x] \\
		&\leq
		\Pr[e \not \in R(x) \;\;\; \forall \;e \in N \setminus \{n\}] \\
		&=
		(1-b)^{n-1}
		<
		c
		\enspace,
  \end{align*}
  }
  {
	\[
		\Pr[I \cup \{n\} \in \mathcal{F}_x \;\;\; \forall \;I\subseteq R(x), I\in \mathcal{F}_x]
		\leq
		\Pr[e \not \in R(x) \;\;\; \forall \;e \in N \setminus \{n\}]
		=
		(1-b)^{n-1}
		<
		c
		\enspace,
	\]
}
	which contradicts the $(b, c)$-selectability of the assumed OCRS.
\end{proof}

The next theorem shows that randomized greedy OCRSs can do much better. Some of the ideas used by this theorem can be traced back to an offline CRS presented by~\cite{chekuri_2014_submodular} for knapsack constraints.

\begin{theorem}
  For every $b \in [0,1/2]$, there exists a $\left(b, \frac{1-2b}{2 - 2b}\right)$-selectable
  randomized greedy OCRS for any polytope $P$ defined by
  a knapsack constraint.
  \label{thm:ocrsknapsack}
\end{theorem}
\begin{proof}
  Let $x \in bP$ be the input point to the OCRS, and let $N_{\mbox{\scriptsize big}}  = \{e\in N \mid s_e > 1/2\}$ be the subset of elements that are big. We use $b_{\mbox{\scriptsize big}}$ to denote the total part of the knapsack occupied by big elements in the fractional solution $x$. Formally,
	\[
		b_{\mbox{\scriptsize big}}
		=
		\sum_{e \in N_{\mbox{\scriptsize big}}} s_e x_e
		\enspace.
	\]
	Observe that $b_{\mbox{\scriptsize big}}$ is always within the range $[0, b]$. The randomized greedy OCRS we use is defined as follows. With probability $p_{\mbox{\scriptsize big}}$ accept greedily the elements of $N_{\mbox{\scriptsize big}}$ while respecting
  the knapsack inequality, where $p_{\mbox{\scriptsize big}}$ is the probability
	\[
		p_{\mbox{\scriptsize big}}
		=
		\frac{1 - 2b + 2b_{\mbox{\scriptsize big}}}{2 - 2b}
		\enspace.
	\]
	With the remaining probability accept greedily the
  small elements of $N \setminus N_{\mbox{\scriptsize big}}$ while respecting the
  knapsack inequality. It is easy to see that this OCRS is indeed a randomized greedy OCRS.

  We continue to analyze the selectability of the above OCRS. Observe that for any big element $e' \in N_{\mbox{\scriptsize big}}$:
  \ifbool{SODACamera}
  {
	\begin{multline} \label{eq:big_knapsack_simplification}
    \Pr[I \cup  \{e'\} \in \mathcal{F}_x \;\;\; \forall \;I\subseteq R(x),I\in \mathcal{F}_x] \\
		 \geq
		\Pr\left[\{e'\} \in \mathcal{F}_x, {\textstyle \sum_{e \in R(x) \cap N_{\mbox{\scriptsize big}}}} s_e \leq 1 - s_{e'}\right]
		\enspace.
	\end{multline}
  }
  {
	\begin{equation} \label{eq:big_knapsack_simplification}
		\Pr[ I \cup \{e'\} \in \mathcal{F}_x \;\;\; \forall \;I\subseteq R(x),I\in \mathcal{F}_x]
		\geq
		\Pr\left[\{e'\} \in \mathcal{F}_x, {\textstyle \sum_{e \in R(x) \cap N_{\mbox{\scriptsize big}}}} s_e \leq 1 - s_{e'}\right]
		\enspace.
	\end{equation}
  }
	The event $\{e'\} \in \mathcal{F}_x$ is simply the event that the OCRS decides to accept big elements. Given that this event occurs, the condition $\sum_{e \in R(x) \cap N_{\mbox{\scriptsize big}}} s_e \leq 1 - s_{e'}$ guarantees that for every subset $I$ of $R(x) \cap N_{\mbox{\scriptsize big}}$ one can add the element $e'$ without violating the knapsack constraint. We can lower bound~\eqref{eq:big_knapsack_simplification} as follows.
  \ifbool{SODACamera}
  {
	\begin{align*}
		\Pr &\left[\{e'\} \in \mathcal{F}_x, {\textstyle \sum_{e \in R(x) \cap N_{\mbox{\scriptsize big}}}} s_e \leq 1 - s_{e'}\right] \\
		& =
    p_{\mbox{\scriptsize big}} \cdot \Pr[R(x) \cap (N_{\mbox{\scriptsize big}} \setminus \{e'\}) = \varnothing]\\
    &\geq 
		\frac{1 - 2b + 2b_{\mbox{\scriptsize big}}}{2 - 2b} \cdot (1- x(N_{\mbox{\scriptsize big}})) \\
		&\geq
		\frac{1 - 2b + 2b_{\mbox{\scriptsize big}}}{2 - 2b} \cdot (1 - 2b_{\mbox{\scriptsize big}})
%		={} &
%		\frac{1 - 2b + 4b_{\mbox{\scriptsize big}}(b - b_{\mbox{\scriptsize big}})}{2 - 2b}
		\geq
		\frac{1 - 2b}{2 - 2b}
		\enspace,
  \end{align*}
  }
  {
	\begin{align*}
		&
		\Pr\left[\{e'\} \in \mathcal{F}_x, {\textstyle \sum_{e \in R(x) \cap N_{\mbox{\scriptsize big}}}} s_e \leq 1 - s_{e'}\right]
		=
    p_{\mbox{\scriptsize big}} \cdot \Pr[R(x) \cap (N_{\mbox{\scriptsize big}} \setminus \{e'\}) = \varnothing]\\
    \geq{} &
		\frac{1 - 2b + 2b_{\mbox{\scriptsize big}}}{2 - 2b} \cdot (1- x(N_{\mbox{\scriptsize big}}))
		\geq
		\frac{1 - 2b + 2b_{\mbox{\scriptsize big}}}{2 - 2b} \cdot (1 - 2b_{\mbox{\scriptsize big}})
%		={} &
%		\frac{1 - 2b + 4b_{\mbox{\scriptsize big}}(b - b_{\mbox{\scriptsize big}})}{2 - 2b}
		\geq
		\frac{1 - 2b}{2 - 2b}
		\enspace,
  \end{align*}
}
	where the first inequality follows from the union bound and the second inequality uses the fact that $s_e > 1/2$ for every $e \in N_{\mbox{\scriptsize big}}$. Similarly, for an element $e\not\in N_{\mbox{\scriptsize big}}$ we have:
  \ifbool{SODACamera}
  {
  \begin{align*}
		\Pr[& I \cup \{e'\} \in \mathcal{F}_x \;\;\; \forall \;I\subseteq R(x),I\in \mathcal{F}_x] \\
		& \geq
		\Pr\left[\{e'\} \in \mathcal{F}_x, {\textstyle \sum_{e \in R(x) \setminus N_{\mbox{\scriptsize big}}}} s_e \leq 1 - s_{e'}\right]\\
		& \geq{} 
    (1 - p_{\mbox{\scriptsize big}}) \cdot \Pr[ s(R(x) \setminus (N_{\mbox{\scriptsize big}} \cup \{e'\})) \leq 1/2] \\
		& \geq
		\left(1 - \frac{1 - 2b + 2b_{\mbox{\scriptsize big}}}{2 - 2b}\right) \cdot \left(1 - \frac{b - b_{\mbox{\scriptsize big}}}{1/2}\right)\\
		& ={} 
		\frac{1 - 2b_{\mbox{\scriptsize big}}}{2 - 2b} \cdot (1 - 2b + 2b_{\mbox{\scriptsize big}})
		\geq
		\frac{1 - 2b}{2 - 2b}
		\enspace,
  \end{align*}
  }
  {
  \begin{align*}
		&
		\Pr[ I \cup \{e'\} \in \mathcal{F}_x \;\;\; \forall \;I\subseteq R(x),I\in \mathcal{F}_x]
		\geq
		\Pr\left[\{e'\} \in \mathcal{F}_x, {\textstyle \sum_{e \in R(x) \setminus N_{\mbox{\scriptsize big}}}} s_e \leq 1 - s_{e'}\right]\\
		\geq{} &
    (1 - p_{\mbox{\scriptsize big}}) \cdot \Pr[ s(R(x) \setminus (N_{\mbox{\scriptsize big}} \cup \{e'\})) \leq 1/2]
		\geq
		\left(1 - \frac{1 - 2b + 2b_{\mbox{\scriptsize big}}}{2 - 2b}\right) \cdot \left(1 - \frac{b - b_{\mbox{\scriptsize big}}}{1/2}\right)\\
		={} &
		\frac{1 - 2b_{\mbox{\scriptsize big}}}{2 - 2b} \cdot (1 - 2b + 2b_{\mbox{\scriptsize big}})
		\geq
		\frac{1 - 2b}{2 - 2b}
		\enspace,
  \end{align*}
}
	where the third inequality follows from Markov's inequality.
\end{proof}

\subsection{Combining constraints\ifbool{SODACamera}{.}{}} \label{ssc:combine}

Like offline CRSs (see~\cite{chekuri_2014_submodular}), greedy OCRSs can be combined to form greedy OCRSs for more involved constraints. 

%Throughout this section we use $\cF_{x, \pi}$ to denote the down-closed family $\cF_x$ used by a greedy OCRS $\pi$ for an input $x$.

\begin{definition}
Given two greedy OCRSs $\pi^1$ and $\pi^2$ for polytopes $P_1$ and $P_2$, the
\emph{combination} of $\pi^1$ and $\pi^2$ is a greedy OCRS $\pi$ for the
polytope $P = P_1 \cap P_2$. For every input $x \in P$, $\pi$ is defined by the
down-closed family $\cF_{\pi, x} = \cF_{\pi^1, x} \cap \cF_{\pi^2, x}$.
\end{definition}

Theorem~\ref{thm:combineOCRSs} is an immediate implication of the next lemma.

\begin{lemma}\label{lem:combineOCRSs}
Let $\pi^1$ and $\pi^2$ be $(b,c_1)$-selectable and $(b, c_2)$-selectable greedy OCRSs, respectively. The combination of $\pi^1$ and $\pi^2$ is $(b,c_1\cdot c_2)$-selectable.
\end{lemma}
\begin{proof}
Let $P_1$ and $P_2$ be the polytopes of $\pi^1$ and $\pi^2$, respectively. Additionally, let $x \in b(P_1 \cap P_2)$, and let $e$ be an arbitrary element of $N$. We need to prove that:
\ifbool{SODACamera}
{
  \begin{multline*}
	\Pr[I \cup \{e\} \in \cF_{\pi^1, x} \cap \cF_{\pi^2, x} \\ \forall \;I\subseteq R(x),  I\in \cF_{\pi^1, x} \cap \cF_{\pi^2, x}] 
	\geq
	c_1 \cdot c_2
	\enspace.
\end{multline*}
}
{
\[
	\Pr[ I \cup \{e\} \in \cF_{\pi^1, x} \cap \cF_{\pi^2, x} \;\;\; \forall \;I\subseteq R(x), I\in \cF_{\pi^1, x} \cap \cF_{\pi^2, x}]
	\geq
	c_1 \cdot c_2
	\enspace.
\end{equation*}
}
For ease of notation, let us denote by $\chi_e(A, \cF, \cF')$ an indicator for the event that $I \cup \{e\} \in \cF$ for every set $I\subseteq A$ obeying $I\in \cF'$. Using this notation, the inequality that we need to prove becomes: $\Pr[\chi_e(R(x), \cF_{\pi^1, x} \cap \cF_{\pi^2, x}, \cF_{\pi^1, x} \cap \cF_{\pi^2, x})] \geq c_1 \cdot c_2$.

On the other hand, observe that $x \in b(P_1 \cap P_2) \subseteq bP_1$. Hence, the $(b, c_1)$-selectability of $\pi^1$ implies that:
\ifbool{SODACamera}
{
\begin{align*}
	\Pr&[\chi_e(R(x), \cF_{\pi^1, x}, \cF_{\pi^1, x} \cap \cF_{\pi^2, x})] \\
	& \geq{}
	\Pr[\chi_e(R(x), \cF_{\pi^1, x}, \cF_{\pi^1, x})]\\
	& ={} 
	\Pr[ I \cup \{e\} \in \cF_{\pi^1, x} \;\;\; \forall \;I\subseteq R(x), I\in \cF_{\pi^1, x}] \\
	& \geq
	c_1
	\enspace.
\end{align*}
}
{
\begin{align*}
	\Pr[\chi_e(R(x), \cF_{\pi^1, x}, \cF_{\pi^1, x} \cap \cF_{\pi^2, x})]
	\geq{} &
	\Pr[\chi_e(R(x), \cF_{\pi^1, x}, \cF_{\pi^1, x})]\\
	={} &
	\Pr[ I \cup \{e\} \in \cF_{\pi^1, x} \;\;\; \forall \;I\subseteq R(x), I\in \cF_{\pi^1, x}]
	\geq
	c_1
	\enspace.
\end{align*}
}
where the first inequality follows since $\chi_e$ is a non-increasing function of its third argument (when the other two arguments are fixed). Similarly, we also get: $\Pr[\chi_e(R(x), \cF_{\pi^2, x}, \cF_{\pi^1, x} \cap \cF_{\pi^2, x})] \geq c_2$. Next, observe that $\chi_e$ is also non-increasing in its first argument (when the other two arguments are fixed). Hence, if we let $R$ be a random set distributed like $R(x)$, then, by the FKG inequality:
\ifbool{SODACamera}
{
\begin{align*}
	&\Pr[\chi_e(R(x), \cF_{\pi^1, x} \cap \cF_{\pi^2, x}, \cF_{\pi^1, x} \cap \cF_{\pi^2, x})] \\
   & =  \Pr[\chi_e(R, \cF_{\pi^1, x}, \cF_{\pi^1, x} \cap \cF_{\pi^2, x}) \\ & \qquad  \cdot \chi_e(R, \cF_{\pi^2, x}, \cF_{\pi^1, x} \cap \cF_{\pi^2, x})] \\
& 	\geq{} 
	\Pr[\chi_e(R(x), \cF_{\pi^1, x}, \cF_{\pi^1, x} \cap \cF_{\pi^2, x})] \\ & \qquad \cdot \Pr[\chi_e(R(x), \cF_{\pi^2, x}, \cF_{\pi^1, x} \cap \cF_{\pi^2, x})] \\
	& \geq
	c_1 \cdot c_2
	\enspace.
\end{align*}
}
{
\begin{align*}
	&
	\Pr[\chi_e(R(x), \cF_{\pi^1, x} \cap \cF_{\pi^2, x}, \cF_{\pi^1, x} \cap \cF_{\pi^2, x})]\\
	={} &
	\Pr[\chi_e(R, \cF_{\pi^1, x}, \cF_{\pi^1, x} \cap \cF_{\pi^2, x}) \cdot \chi_e(R, \cF_{\pi^2, x}, \cF_{\pi^1, x} \cap \cF_{\pi^2, x})]\\
	\geq{} &
	\Pr[\chi_e(R(x), \cF_{\pi^1, x}, \cF_{\pi^1, x} \cap \cF_{\pi^2, x})] \cdot \Pr[\chi_e(R(x), \cF_{\pi^2, x}, \cF_{\pi^1, x} \cap \cF_{\pi^2, x})]
	\geq
	c_1 \cdot c_2
	\enspace.
  \ifbool{SODACamera}{}{\qedhere}
\end{align*}
}
\end{proof}
}%end ifbool{shortVersion}

%% file: SubmodularResults.tex
\section{From Selectability to Approximation} \label{sec:selectability_to_approximation}

In this section we prove Theorems~\ref{thm:submodular_basic} and~\ref{thm:non_negative_special} which study the relation between the value of a vector $x$ and the the expected value of the output produced by an OCRS given this vector as input. Following is a restatement of Theorem~\ref{thm:submodular_basic}. Unlike the original statement of the theorem, this restatement uses the notation $S(p)$, where $S$ is a set and $p$ is a probability, to denote a random set containing every element $e \in S$ with probability $p$, independently. Observe that $S(p)$ has the same distribution as $R(p \cdot \characteristic_S)$.

\begin{reptheorem}{thm:submodular_basic}
Given a non-negative monotone submodular function $f\colon 2^N \to \bR_{\geq 0}$ and a $(b, c)$-selectable greedy OCRS for a polytope $P$, applying the greedy OCRS to an input $x \in bP$ results in a random set $S$ satisfying $\bE[f(S)] \geq c \cdot F(x)$, where $F$ is the multilinear extension of $f$. Moreover, even if $f$ is not monotone, $\bE[f(S(1/2))] \geq (c/4) \cdot F(x)$, where the random decisions used to calculate $S(1/2)$ are considered part of the algorithm, and thus, known to the almighty adversary.
\end{reptheorem}

To prove Theorem~\ref{thm:submodular_basic} we need to define some concepts regarding \emph{offline} CRSs. Recall that a CRS for a polytope $P$ is a (possibly random) function $\pi\colon 2^N \to 2^N$ that depends on an input vector $x \in P$, and for every set $S \subseteq N$ returns a subset $\pi(S) \subseteq S$ that obeys the polytope $P$ (\ie, $\characteristic_{\pi(S)} \in P$). 

\begin{definition}
A CRS $\pi$ for a polytope $P$ is:
\begin{itemize}\setlength{\itemsep}{0pt}
	\item \emph{$(b, c)$-balanced} some for $b, c \in [0, 1]$ if $\Pr[e \in \pi(R(x)) \mid e \in R(x)] \geq c \cdot x_e$ whenever the input vector $x$ belongs to $bP$ and $x_e > 0$.
	\item \emph{monotone} if $\Pr[e \in \pi(S_1)] \geq \Pr[e \in \pi(S_2)]$ whenever $e \in S_1 \subseteq S_2$.
\end{itemize}
\end{definition}

We now define and analyze an interesting notion that allows us to connect the world of OCRSs with that of CRSs.

\begin{definition}[characteristic CRS]
The \emph{characteristic CRS} $\bar{\pi}$ of a greedy OCRS $\pi$ for a polytope $P$ is a CRS for the same polytope $P$. It is defined for an input $x \in P$ and a set $A \subseteq N$ by $\bar{\pi}(A) = \{e \in A \mid I \cup \{e\} \in \cF_{\pi, x}\;\; \forall I \subseteq A, I \in \cF_{\pi, x}\}$.
\end{definition}

The following observation shows that $\bar{\pi}(A)$ obeys the polytope $P$ and $\bar{\pi}(A) \subseteq A$. Thus, the characteristic CRS $\bar{\pi}$ is a true CRS for the polytope $P$.

\begin{observation} \label{obs:subset}
For every set $A \subseteq N$ and a characteristic CRS $\bar{\pi}$ of a greedy OCRS $\pi$, the set $\bar{\pi}(A)$ is always a subset of the elements selected by $\pi$ when the active elements are the elements of $A$.
\end{observation}
\begin{proof}
Fix an element $e \in \bar{\pi}(A)$, and let $T$ be the set of elements selected by the greedy OCRS $\pi$ immediately before $e$ reveals whether it is active. The definition of $\bar{\pi}(A)$ guarantees that $T \cup \{e\} \in \cF_{\pi, x}$ since $T \in \cF_{\pi, x}$, and thus, $e$ is accepted by $\pi$.
\end{proof}

\begin{lemma} \label{lem:characteristic}
The characteristic CRS $\bar{\pi}$ of a $(b, c)$-selectable greedy OCRS $\pi$ is $(b, c)$-balanced and monotone.
\end{lemma}
\begin{proof}
Let $P$ denote the polytope of the greedy OCRS $\pi$ (and its characteristic CRS $\bar{\pi}$), and fix a vector $x \in bP$. Since $\pi$ is $(b, c)$-selectable, we get for a set $A$ distributed like $R(x)$ and an arbitrary element $e \in N$:
\ifbool{SODACamera}
{
\begin{align*}
	\Pr& [e \in \bar{\pi}(A) \mid e \in A] \\
	 & =
	\Pr[I \cup \{e\} \in \cF_{\pi,x} \;\;\; \forall \;I\subseteq A, I\in \mathcal{F}_{\pi,x} \mid e \in A]\\
	& =
	\Pr[I \cup \{e\} \in \cF_{\pi,x} \;\;\; \forall \;I\subseteq A, I\in \mathcal{F}_{\pi,x}]
	\geq
	c
	\enspace.
\end{align*}
}
{
\begin{align*}
	\Pr[e \in \bar{\pi}(A) \mid e \in A]
	={} &
	\Pr[I \cup \{e\} \in \cF_{\pi,x} \;\;\; \forall \;I\subseteq A, I\in \mathcal{F}_{\pi,x} \mid e \in A]\\
	={} &
	\Pr[I \cup \{e\} \in \cF_{\pi,x} \;\;\; \forall \;I\subseteq A, I\in \mathcal{F}_{\pi,x}]
	\geq
	c
	\enspace.
\end{align*}
}
The last inequality implies, by definition, that $\bar{\pi}$ is a $(b,c)$-balanced CRS.

Next, let us prove that $\bar{\pi}$ is monotone. Fix an instantiation of $\cF_{\pi, x}$, an element $e \in N$ and two sets $e \in A_1 \subseteq A_2 \subseteq N$. If $e \in \bar{\pi}(A_2)$, then we know that $I \cup \{e\} \in \cF_{\pi, x}$ for every $I\subseteq A_2, I\in \mathcal{F}_{\pi, x}$. Thus, clearly, $I \cup \{e\} \in \cF_{\pi, x}$ for every $I\subseteq A_1 \subseteq A_2, I\in \mathcal{F}_{\pi, x}$, which implies $e \in \bar{\pi}(A_1)$. In summary, we got that:
\[
	e \in \bar{\pi}(A_2)
	\Rightarrow
	e \in \bar{\pi}(A_1)
	\enspace,
\]
and thus, even when we do not fix the instantiation of $\cF_{\pi, x}$:
\[
	\Pr[e \in \bar{\pi}(A_2)] \leq \Pr[e \in \bar{\pi}(A_1)]
	\enspace,
\]
which completes the proof that $\bar{\pi}$ is monotone.
\end{proof}

To continue the proof of Theorem~\ref{thm:submodular_basic} we need to state some known results. The following lemma follows from the work of~\cite{chekuri_2014_submodular} (mostly from Theorem~1.3 in that work).

\begin{lemma} \label{lem:offline_crs}
For every given non-negative submodular function $f\colon 2^N \to \bR_{\geq 0}$, there exists a function $\eta_f \colon 2^N \to 2^N$ that always returns a subset of its argument (\ie, $\eta_f(S) \subseteq S$ for every $S \subseteq N$) having the following property. For every monotone $(b, c)$-balanced CRS $\pi$ for a polytope $P$ and input vector $x \in bP$:
\[
	\bE[f(\eta_f(\pi(R(x))))]
	\geq
	c \cdot F(x)
	\enspace,
\]
where $F(x)$ is the multilinear extension of $f$.
\end{lemma}

The following two known lemmata about submodular functions have been rephrased to make them fit our notation better.

\begin{lemma}[Lemma~2.2 of~\cite{Feige_2011_non-monotone}] \label{lem:sampling}
Let $g\colon 2^N \to \bR$ be a submodular function, and let $T$ be an arbitrary set $T \subseteq N$. For every random set $T_p \subseteq T$ which contains every element of $T$ with probability $p$ (not necessarily independently):
\[
	\bE[g(T_p)] \geq (1 - p) g(\varnothing) + p \cdot g(T)
	\enspace.
\]
\end{lemma}

\begin{lemma}[Lemma~2.2 of~\cite{buchbinder_2014_cardinality}] \label{lem:max_prob_bound}
Let $g\colon 2^N \to \bR_{\geq 0}$ be a non-negative submodular function.
For every random set $N_p\subseteq N$ which contains every element
of $N$ with probability at most $p$ (not necessarily independently):
\[
	\bE[g(N_p)] \geq (1 - p)g(\varnothing)
	\enspace.
\]
\end{lemma}

We are now ready to prove Theorem~\ref{thm:submodular_basic}.
\begin{proof}[Proof of Theorem~\ref{thm:submodular_basic}]
Let $\bar{\pi}$ be the characteristic CRS of the OCRS $\pi$ we consider, and let $A \sim R(x)$ be the set of active elements. For notational convenience, let us also denote by $S'$ the set $\eta_f(\bar{\pi}(A))$. Notice that $S'$ is a subset of the set $S$ of accepted elements since $\bar{\pi}(A)$ is a subset of $S$ by Observation~\ref{obs:subset}. Lemmata~\ref{lem:characteristic} and~\ref{lem:offline_crs} imply together the inequality $\bE[f(S')] = \bE[f(\eta_f(\bar{\pi}(A)))] \geq c \cdot F(x)$. To complete the proof of the first part of the theorem, it is now enough to observe that by the monotonicity of $f$: $\bE[f(S)] \geq \bE[f(S')] \geq c \cdot F(x)$.

To prove the second part of the theorem, let us fix the set $A$ and the family $\cF_{\pi, x}$ which, for brevity, we denote by $\mathcal{F}_x$ in the rest of the proof. Observe that the set $S'$ is deterministic once $A$ and $\cF_{x}$ are fixed, and let us denote this set by $S'_{A, \cF_{x}}$. Hence, we can think of $S(1/2)$ as obtained by first calculating a set $S'_{A, \cF_{x}}(1/2)$ containing every element of $S'_{A, \cF_{x}}$ with probability $1/2$, independently, and then adding to it a random set $\Delta \subseteq N \setminus S'_{A, \cF_{x}}$. By controlling the order in which elements reveal whether they are active, the adversary can make the distribution of $\Delta$ depend on $S'_{A, \cF_{x}}(1/2)$; however, $\Delta$ is guaranteed to contain every element with probability at most $1/2$ for every given choice of $S'_{A, \cF_{x}}(1/2)$. Using this observation we get:
\ifbool{SODACamera}
{
\allowdisplaybreaks
\begin{align*}
	\bE&[f(S(1/2)) \mid A, \cF_x]
	=
	\bE[f(S'(1/2) \cup \Delta) \mid A, \cF_x]\\
	={} &
	\sum_{B \subseteq S'_{A, \cF_x}} \mspace{-18mu}
     \Pr[S'_{A, \cF_x}(1/2) = B \mid A, \cF_x] \\ &\cdot \bE[f(B \cup \Delta) \mid S'_{A, \cF_x}(1/2) = B, A, \cF_x]\\
	\geq{} &
	\sum_{B \subseteq S'_{A, \cF_x}} \mspace{-18mu}
     \Pr[S'_{A, \cF_x}(1/2) = B \mid A, \cF_x] \\ &\cdot \frac{\bE[f(B) \mid A, \cF_x]}{2} \\
	={} &
	\frac{\bE[f(S'(1/2)) \mid A, \cF_x]}{2}
	\geq{} 
	\frac{\bE[f(S') \mid A, \cF_x]}{4}
	\enspace,
\end{align*}
}
{
\begin{align*}
	\bE[f(S(1/2)) \mid A, \cF_x]
	={} &
	\bE[f(S'(1/2) \cup \Delta) \mid A, \cF_x]\\
	={} &
	\sum_{B \subseteq S'_{A, \cF_x}} \mspace{-18mu}
     \Pr[S'_{A, \cF_x}(1/2) = B \mid A, \cF_x] \cdot \bE[f(B \cup \Delta) \mid S'_{A, \cF_x}(1/2) = B, A, \cF_x]\\
	\geq{} &
	\sum_{B \subseteq S'_{A, \cF_x}} \mspace{-18mu}
     \Pr[S'_{A, \cF_x}(1/2) = B \mid A, \cF_x] \cdot \frac{\bE[f(B) \mid A, \cF_x]}{2}
	=
	\frac{\bE[f(S'(1/2)) \mid A, \cF_x]}{2}\\
	\geq{} &
	\frac{\bE[f(S') \mid A, \cF_x]}{4}
	\enspace,
\end{align*}
}
where the first inequality follows from Lemma~\ref{lem:max_prob_bound} since the function $h_B(T) = f(B \cup T)$ is non-negative and submodular for every set $B \subseteq N$, and the second inequality follows from Lemma~\ref{lem:sampling}. Taking an expectation over the possible values of the set $A$ and the family $\cF_x$, we get:
\ifbool{SODACamera}{
\begin{align*}
	\bE[f(S(1/2))]
	={} &
	\bE_{A, \cF_x}[\bE[f(S(1/2)) \mid A, \cF_x]]\\
	\geq{} &
	\bE_{A, \cF_x}\left[\frac{\bE[f(S') \mid A, \cF_x]}{4}\right]\\
  ={} &
	\frac{\bE[f(S')]}{4}
	\geq
	(c/4) \cdot F(x)
	\enspace.
  \ifbool{SODACamera}{}{\qedhere}
\end{align*}
}
{
\begin{align*}
	\bE[f(S(1/2))]
	={} &
	\bE_{A, \cF_x}[\bE[f(S(1/2)) \mid A, \cF_x]]\\
	\geq{} &
	\bE_{A, \cF_x}\left[\frac{\bE[f(S') \mid A, \cF_x]}{4}\right]
	=
	\frac{\bE[f(S')]}{4}
	\geq
	(c/4) \cdot F(x)
	\enspace.
  \ifbool{SODACamera}{}{\qedhere}
\end{align*}
}
\end{proof}

\ifbool{SODACamera}
{
  Notice that the result proved by Theorem~\ref{thm:submodular_basic} for
  non-monotone functions loses a factor of $4$ in the guarantee. To avoid that,
  we also consider online adversaries. In the full version of this paper~\cite{feldman_2015_online} we show that for such adversaries one can prove a variant of the above result that does not lose a factor $4$ for non-monotone functions.
}
{
Notice that the result proved by Theorem~\ref{thm:submodular_basic} for non-monotone functions loses a factor of $4$ in the guarantee. To avoid that, we also consider online adversaries. 
%\begin{definition}[history adaptive adversary]
%A history adaptive adversary is an adversary that determines on the fly the order in which elements report whether they are active. Whenever the adversary needs to choose the next element that will report its activity state, it chooses an element based only on information available to the algorithm at this point. More specifically, for an element that already revealed its activity state the adversary knows whether it is active and whether the algorithm selected it. However, the adversary has no information about the activity state of elements that did not reveal it yet.
%\end{definition}
Unfortunately, we do not have an improved result for greedy OCRSs against online adversaries. Instead, we study the performance of a different family of ORCSs against such adversaries.

\begin{definition}[element-monotone OCRS]
An element-monotone OCRS is an OCRS characterized by a down-closed family $\cF_u \subseteq 2^{N \setminus \{e\}}$ for every element $e \in N$, where the families $\{\cF_e\}_{e \in N}$ can be either deterministic (a deterministic element-monotone OCRS) or taken from some joint distribution (a randomized element-monotone OCRS). Such an OCRS accepts an active element $e$ if and only if $A^{<e} \in \cF_e$, where $A^{<e}$ is the set of active elements that revealed that they are active before $e$ does so.
\end{definition}

One can observe that all the results we prove for greedy OCRSs in this paper can be easily applied also to element-monotone OCRSs with the following modified definition of $(b, c)$-selectability.
\begin{definition}[$(b,c)$-selectability for element-monotone OCRSs]
Let $b,c \in [0,1]$. An element-monotone OCRS for $P$ is $(b,c)$-selectable if for any $x\in b\cdot P$ we have
\[
	\Pr[R(x) \setminus \{e\} \in \cF_e] \geq c
  \qquad \forall e\in N
	\enspace.
\]
\end{definition}
\noindent Moreover, the greedy OCRSs we describe for specific polytopes can be converted into similar element-monotone OCRSs with the same $(b, c)$-selectability guarantee (in fact the greedy OCRS we describe for matroids is already an element-monotone OCRS without any modifications). Our objective in the rest of this section is to prove the next theorem. Notice that Theorem~\ref{thm:non_negative_special} is a special case of this theorem.

\begin{theorem} \label{thm:element_monotone_non_monotone_function}
Given an element-monotone $(b, c)$-selectable OCRS $\pi$ and a non-negative submodular function $f\colon 2^N \to \bR_{\geq 0}$, there exists an OCRS $\pi'$ that for every input vector $x \in bP$ and online adversary selects a random set $S$ such that $\bE[f(S))] \geq c \cdot F(x)$.
\end{theorem}

To prove Theorem~\ref{thm:element_monotone_non_monotone_function}, we need some notation. Let $\sigma_a$ denote an arbitrary fixed absolute order over the elements of $N$. Given element $e \in N$ and a vector $y \in [0, 1]^N$, let $F(e \mid y)$ denote the marginal contribution obtained by increasing the coordinate of $e$ in $y$ to $1$. Formally, $F(e \mid y) = F(y \vee \characteristic_{\{e\}}) - F(y)$. Finally, we say that $x$ is \emph{non-reducible} if $F(e \mid x \wedge \characteristic_{\sigma_a^{<e}}) \geq 0$ for every element $e \in \cN$ having $x_e > 0$, where $\sigma_a^{<e}$ is the set of elements that appear before $e$ in the order $\sigma_a$.

\begin{observation} \label{obs:to_non_reducible}
If $x$ is reducible, then there exists a non-reducible vector $x' \leq x$ obeying $F(x') \geq F(x)$.
\end{observation}
\begin{proof}
Consider the vector $x'$ obtained from $x$ by the following process. Start with $x' \gets x$. Scan the elements in the order $\sigma_a$. For every element $e$, if $F(e \mid x' \wedge \characteristic_{\sigma_a^{<e}}) < 0$, then reduce $x'_e$ to $0$. Clearly this process ends up with a non-reducible vector $x'$. Moreover, every step of the process only increases the value of $F(x')$, and thus, this value ends up at least as large as its initial value $F(x)$. To see why this is the case, consider a step in which the value of $x'_e$ is reduced to $0$ for some element $e \in N$, and let $x^1$ and $x^2$ be the vector $x'$ before and after the reduction. Then,
\begin{align*}
	F(x^1)
	={} &
	f(\varnothing) + \sum_{e' \in N} x^1_{e'} \cdot F(e' \mid x^1 \wedge \characteristic_{\sigma_a^{<e'}})
	\leq
	f(\varnothing) + \sum_{e' \in N \setminus \{e\}} x^1_{e'} \cdot F(e' \mid x^1 \wedge \characteristic_{\sigma_a^{<e'}})\\
	\leq{} &
	f(\varnothing) + \sum_{e' \in N \setminus \{e\}} x^1_{e'} \cdot F(e' \mid x^2 \wedge \characteristic_{\sigma_a^{<e'}})
	=
	F(x^2)
	\enspace,
\end{align*}
where the first inequality holds since $F(e \mid x^1 \wedge \characteristic_{\sigma_a^{<e}}) < 0$, and the second inequality holds by submodularity.
\end{proof}

Observation~\ref{obs:to_non_reducible} shows that it is enough to prove Theorem~\ref{thm:element_monotone_non_monotone_function} for non-reducible $x$. If $x$ is reducible, the OCRS $\pi'$ can calculate a non-reducible $x'$ having $F(x') \geq F(x)$, and then ``pretend'' that some active elements are in fact inactive in a way that makes every element $e$ active with probability $x'_e$, independently. Determining $x'$ requires exponential time, but even if one wants a polynomial time OCRS $\pi'$ it is possible to use sampling to get, for every constant $d > 0$, a non-reducible vector $x'$ having $F(x') \geq F(x) - |N|^{-d} \cdot \max\{f(\{e\}) \mid x_e > 0\}$.

The OCRS $\pi'$ we use in the proof of Theorem~\ref{thm:element_monotone_non_monotone_function} works as follows: whenever $\pi'$ learns that an element $e$ is active it complements the set $A^{<e}$ to have the distribution $R(x \wedge \characteristic_{N \setminus \{e\}})$, and then checks whether the resulting random set is in $\cF_e$. More formally, let $\sigma$ be the order in which the elements reveal whether they are active, and let $\sigma^{>e}$ be the set of elements that appears after $e$ in the order $\sigma$. Then, when $e$ reveals that it is active, the OCRS $\pi'$ decides to accept it if $A^{<e} \cup R(x \wedge \characteristic_{\sigma^{>e}}) \in \cF_e$ (for a random realization of the random set $R(x \wedge \characteristic_{\sigma^{>e}})$.

\begin{observation}
The OCRS $\pi'$ always selects a set $S \in \cF$.
\end{observation}
\begin{proof}
Fix the set $A$ of active elements, the families $\{\cF_e\}_{e \in N}$ and the order $\sigma$ in which elements reveal whether they are active. We prove that the observation holds for any possible choice for these fixed values.

Let $S'$ be the set of elements selected by $\pi$. The fact that $\pi'$ selects some element $e \in N$ (\ie, $e \in S$) means that $(S \cap \sigma^{<e}) \cup R(x \wedge \characteristic_{\sigma^{>e}}) \in \cF_e$ for some realization of the random set $R(x \wedge \characteristic_{\sigma^{>e}})$. Hence, $S \cap \sigma^{<e} \in \cF_e$, which implies that $\pi$ also selects $e$ (\ie, $e \in S'$). This argument implies that $S \subseteq S'$. The observation now follows by the down-monotonicity of $\cF$ since $\pi$ always returns a set in $\cF$.
\end{proof}

Next, we need to prove the following technical lemma.

\begin{lemma} \label{lem:dependence_better}
Let $X_1, X_2$ be two random sets that never contain elements $e, e' \in N$. Additionally, let $Y_1$ and $Y_2$ be two random sets distributed like $R(x \wedge \characteristic_{\{e'\}})$ which are independent from each other and from $X_1$, $X_2$ and $\cF_e$. Then, assuming $\Pr[X_2 \in \cF_e] > 0$,
\[
	\bE[f(e \mid X_1 \cup Y_1) \mid X_2 \cup Y_2 \in \cF_e]
	\leq
	\bE[f(e \mid X_1 \cup Y_1) \mid X_2 \cup Y_1 \in \cF_e]
	\enspace.
\]
\end{lemma}
\begin{proof}
It is enough to prove the lemma for fixed values of $X_1$, $X_2$ and $\cF_e$ obeying $X_2 \in \cF_e$. If $X_2 \cup \{e'\} \in \cF_e$, then the conditions $X_2 \cup Y_2 \in \cF_e$ and  $X_2 \cup Y_1 \in \cF_e$ always hold and the lemma is trivial, thus, we also assume $X_2 \cup \{e'\} \not \in \cF_e$. Then, conditioned on the values of $X_1$, $X_2$ and $\cF_e$:
\begin{align*}
	\bE[f(e \mid X_1 \cup Y_1) \mid X_2 \cup Y_2 \in \cF_e]
	\leq{} &
	\bE[f(e \mid X_1) \mid X_2 \cup Y_2 \in \cF_e]
	=
	\bE[f(e \mid X_1) \mid Y_2 = \varnothing]\\
	={} &
	\bE[f(e \mid X_1 \cup Y_1) \mid Y_1 = \varnothing]\\
	={} &
	\bE[f(e \mid X_1 \cup Y_1) \mid X_2 \cup Y_1 \in \cF_e]
	\enspace,
\end{align*}
where the inequality holds by submodularity.
\end{proof}

Using the last lemma we can now prove the following one, which lower bounds the contribution of an element $e$ to the value of the solution selected by $\pi'$.

\begin{lemma} \label{lem:history_adversary}
Let $A \sim R(x)$ be the set of active elements, and let $S$ be the output of $\pi'$ given $A$ and an arbitrary online adversary. Then, for every element $e \in N$ having $x_e > 0$ and $\Pr[\varnothing \in \cF_e] > 0$, $\Pr[e \in S \mid e \in A] = \Pr[R(x) \setminus \{e\} \in \cF_e]$ and $\bE[f(e \mid A \cap \sigma_a^{<e}) \mid e \in S] \geq F(e \mid x \wedge \characteristic_{\sigma_a^{<e}})$.
\end{lemma}
\begin{proof}
Observe that the conditions $x_e > 0$ and $\Pr[\varnothing \in \cF_e] > 0$ are equivalent to $\Pr[e \in S] > 0$, and thus, the second expectation we want to bound is well defined for every element $e$ having $x_e >0$ and $\Pr[\varnothing \in \cF_e] > 0$.

The strategy of an online adversary for determining the order in which elements reveal whether they are active, up to the point when the element $e$ does so, can be characterized by a binary tree with the following properties:
\begin{compactitem}
	\item Every leaf is marked by $e$.
	\item Every internal node is marked by some other element $e' \neq e$ and has two children called the ``accepted child'' and the ``rejected child''.
	\item The path from the root of the tree to every leaf does not contain a single element more than once.
\end{compactitem}
The semantics of understanding the tree as a strategy are as follows. The adversary starts at the root of the tree, and moves slowly down the tree till reaching a leaf. When the adversary is at some node $e'$ it makes $e'$ the next element that reveals whether it is active. If $e'$ is accepted by the OCRS, then the adversary moves to its ``accepted child'', otherwise it moves to its ``rejected child''. Notice that strategies defined by such trees do not allow the adversary to use the information whether $e'$ was active or not when the OCRS rejects it, however, since the activation of elements is independent and our OCRS behaves in the same way under both cases, this information cannot help the adversary.

In the rest of this proof, we show by induction on the number of leafs in the strategy of the adversary that the lemma holds for every online adversary. Let us start with the case of the single strategy having a single leaf. In this strategy the adversary makes $e$ the first element to reveal whether it is active. This means that:
\[
	\Pr[e \in S \mid e \in A] = \Pr[R(x) \setminus \{e\} \in \cF_e]
	\enspace,
\]
and
\begin{align*}
	\bE[f(e \mid A \cap \sigma_a^{<e}) \mid e \in S]
	={} &
	\bE[f(e \mid A \cap \sigma_a^{<e}) \mid R(x) \setminus \{e\} \in \cF_e, e \in A]\\
	={} &
	\bE[f(e \mid A \cap \sigma_a^{<e})]
	=
	F(e \mid x \wedge \characteristic_{\sigma_a^{<e}})
	\enspace,
\end{align*}
where the second equality holds since $A \cap \sigma_a^{<e}$ is independent of the conditions $R(x) \setminus \{e\} \in \cF_e$ and $e \in A$.

Next, assume the lemma holds for strategies having $\ell - 1 \geq 1$ leaves, and let us prove it for a strategy $\cT$ having $\ell$ leaves. Since every internal node of $\cT$ has two children, $\cT$ must contain a node $w$ (marked by an element $e' \neq e$) having two leaves as children. Let $\cT'$ be the strategy resulting from $\cT$ by removing the two children of $w$ and making $w$ itself a leaf (by marking it with $e$). Since $\cT'$ has $\ell - 1$ leaves, it obeys the lemma by the induction hypothesis. Hence, it is enough to show that $\Pr[e \in S \mid e \in A]$ is identical under both $\cT$ and $\cT'$ and $\bE[f(e \mid A \cap \sigma_a^{<e}) \mid e \in S]$ is at least as large under $\cT$ as under $\cT'$.

Let $\chi_w$ be the event that the adversary reaches the node $w$ in its strategy. Clearly $\Pr[\chi_w \mid \cT] = \Pr[\chi_w \mid \cT']$, where by conditioning a probability on a strategy of the adversary we mean that the probability is calculated for the case that the adversary uses this strategy. If the last two probabilities are strictly smaller than $1$ we also have:
\[
	\Pr[e\in S \mid e \in A, \neg \chi_w, \cT]
	=
	\Pr[e\in S \mid e \in A, \neg \chi_w, \cT']
	\enspace.
\]
On the other hand, observe that for both strategies $\cT$ and $\cT'$ the event $\chi_w$ implies the same (deterministic) fixed values for the set $B_w$ of the elements that revealed whether they are active up to the point that the adversary reached $w$, and the set $A_w = A \cap B_w$. Thus, whenever $\Pr[\chi_w \mid \cT] = \Pr[\chi_w \mid \cT'] > 0$:
\begin{align*}
	&
	\Pr[e\in S \mid e \in A, \chi_w, \cT']
	=
	\Pr[A_w \cup (R(x) \setminus (B_w \cup \{e\})) \in \cF_e]\\
	={} &
	\Pr[A_w \cup (A \cap \{e'\}) \cup (R(x) \setminus (B_w \cup \{e, e'\})) \in \cF_e]
	=
	\Pr[e\in S \mid e \in A, \chi_w, \cT]
	\enspace.
\end{align*}
In conclusion, we got: $\Pr[e\in S \mid e \in A, \cT] = \Pr[e\in S \mid e \in A, \cT']$. Notice that both probabilities must be strictly positive by the induction hypothesis.

By Bayes' law we now get:
\begin{align*}
	&
	\Pr[\chi_w \mid \cT, e \in S]
	=
	\Pr[\chi_w \mid e \in A, \cT, e \in S]
	=
	\Pr[e \in S \mid e \in A, \cT, \chi_w] \cdot \frac{\Pr[\chi_w \mid e \in A, \cT]}{\Pr[e \in S \mid e \in A, \cT]}\\
	={} &
	\Pr[e \in S \mid e \in A, \cT', \chi_w] \cdot \frac{\Pr[\chi_w \mid e \in A, \cT']}{\Pr[e \in S \mid e \in A, \cT']}
	=
	\Pr[\chi_w \mid e \in A, \cT', e \in S]
	=
	\Pr[\chi_w \mid \cT', e \in S]
	\enspace,
\end{align*}
and when the probabilities $\Pr[\chi_w \mid \cT, e \in S] = \Pr[\chi_w \mid \cT', e \in S]$ are strictly smaller than $1$ we also have:
\[
	\bE[f(e \mid A \cap \sigma_a^{<e}) \mid e \in S, \neg \chi_w, \cT]
	=
	\bE[f(e \mid A \cap \sigma_a^{<e}) \mid e \in S, \neg \chi_w, \cT']
	\enspace.
\]
Hence, to prove that $\bE[f(e \mid A \cap \sigma_a^{<e}) \mid e \in S]$ is at least as large under $\cT$ as under $\cT'$ we are only left to show the inequality:
\[
	\bE[f(e \mid A \cap \sigma_a^{<e}) \mid e \in S, \chi_w, \cT]
	\geq
	\bE[f(e \mid A \cap \sigma_a^{<e}) \mid e \in S, \chi_w, \cT']
\]
whenever $\Pr[\chi_w \mid \cT, e \in S] > 0$. The last inequality holds since:
\begin{align*}
	&
	\bE[f(e \mid A \cap \sigma_a^{<e}) \mid e \in S, \chi_w, \cT']\\
	={} &
	\bE\Big[f(e \mid (A_w \cup (A \setminus B_w)) \cap \sigma_a^{<e})
    \;\Big\vert\; A_w \cup ((R(x) \setminus B_w) \setminus \{e\}) \in \cF_e\Big]\\
	%={} &
	%\bE[f(e \mid (A_w \cup (R(x) \cap \{e'\}) \cup ((A \setminus B_w) \setminus \{e'\})) \cap \sigma_a^{<e}) \mid A_w \cup \\& (A \cap \{e'\}) \cup ((R(x) \setminus B_w) \setminus \{e', e\}) \in \cF_e]\\
	\leq{} &
	\bE\Big[f(e \mid (A_w \cup (A \setminus B_w)) \cap \sigma_a^{<e}) \;\Big\vert\; A_w \cup (A \cap \{e'\}) \cup ((R(x) \setminus B_w) \setminus \{e', e\}) \in \cF_e\Big]\\
	={} &
	\bE[f(e \mid A \cap \sigma_a^{<e}) \mid e \in S, \chi_w, \cT]
	\enspace,
\end{align*}
where the inequality holds by Lemma~\ref{lem:dependence_better}.
\end{proof}

We are now ready to prove Theorem~\ref{thm:element_monotone_non_monotone_function}.
\begin{proof}[Proof of Theorem~\ref{thm:element_monotone_non_monotone_function}]
If $c = 0$, then the theorem is trivial. Thus, we may assume $c > 0$. Recall that $S$ is the set produced by the OCRS $\pi'$, and let $A \sim R(x)$ be the set of active elements. Since, $\pi$ is a $(b, c)$-selectable element-monotone contention resolution scheme, $A \setminus \{e\}$ must belong to $\cF_e$, for every element $e \in N$, with a positive probability. By the down-monotonicity of $\cF_e$, we get that $\varnothing \in \cF_e$ for every $e \in N$.

Next, observe that:
\begin{align*}
	\bE[f(S)]
	={} &
	f(\varnothing) + \sum_{e \in N, x_e > 0} \mspace{-9mu} \Pr[e \in S] \cdot \bE[f(e \mid S \cap \sigma_a^{<e}) \mid e \in S]\\
	\geq{} &
	f(\varnothing) + \sum_{e \in N, x_e > 0} \mspace{-9mu} \Pr[e \in S] \cdot \bE[f(e \mid A \cap \sigma_a^{<e}) \mid e \in S]\\
	\geq{} &
	f(\varnothing) + \sum_{e \in N, x_e > 0} \mspace{-9mu} \Pr[R(x) \setminus \{e\} \in \cF_e,  e\in A] \cdot F(e \mid x \wedge \characteristic_{\sigma_a^{<e}})\\
  ={} &
	f(\varnothing) + \sum_{e \in N, x_e > 0} \mspace{-9mu} x_e \cdot \Pr[R(x) \setminus \{e\} \in \cF_e] \cdot F(e \mid x \wedge \characteristic_{\sigma_a^{<e}})
	\enspace.
\end{align*}
where the first inequality holds by submodularity and the second by Lemma~\ref{lem:history_adversary}. Since $\pi$ is a $(b, c)$-selectable contention resolution scheme, the probability $\Pr[R(x) \setminus \{e\} \in \cF_e]$ must be at least $c$ for every element $e \in N$. Additionally, since we assumed that $x$ is non-reducible, we also have $F(e \mid x \cap \characteristic_{\sigma_a^{<e}}) \geq 0$ for every element $e \in N$ obeying $x_e > 0$. Plugging both observations into the previous inequality gives:
\[
	\bE[f(S)]
	\geq
	f(\varnothing) + \sum_{e \in N, x_e > 0} \mspace{-9mu} c \cdot x_e\cdot F(e \mid x \wedge \characteristic_{\sigma_a^{<e}})
  = c\cdot F(x) + (1-c)\cdot f(\varnothing)
	\geq
	c \cdot F(x)
	\enspace.
  \ifbool{SODACamera}{}{\qedhere}
\]
\end{proof}
}

%% file: applications.tex
\section{Details on applications}\label{sec:app}

In this section we prove the theorems stated in
Section~\ref{subsec:app} and provide some additional
information. We reuse notation introduced
in Section~\ref{subsec:app}.

\subsection{Prophet inequalities for Bayesian
online selection problems}\label{subsec:appProphet}

\subsubsection*{Proof of Theorem~\ref{thm:prophet}\ifbool{SODACamera}{.}{}}

We start by introducing a relaxation for the expected
value of the prophet, \ie,
$\E[\max\{\sum_{e\in S} Z_e \mid S\in\mathcal{F}\}]$.
Similar relaxation techniques to the one we use here
have been been used previously for different
related problems (see, \eg,~\cite{yan_2010_mechanism}).

Let $I^* \in \argmax\{\sum_{e\in I}Z_i \mid I\in \mathcal{F} \}$
be an optimal (random) set for the prophet, and we define
$p^*_e = \Pr[e\in I^*]$ for $e\in N$.
Our relaxation seeks probabilities $p_e$ that
have the role of $p^*_e$.
Observe first that since $I^*\in \mathcal{F}$ with probability
$1$, we have that $p^*$ is a convex combination
of characteristic vectors of feasible sets, and hence,
$p^*\in P$ (we recall that $P$ is a relaxation for
$\mathcal{F}$).

Our relaxation assigns an optimistic objective value
to each probability vector $p\in P$. More precisely,
for any $e\in N$ we assume in the relaxation
that element $e$ gets selected when $Z_e$ takes
one of its $p_e$-fraction of highest values.
In particular, if $Z_e$ follows a continuous distribution
with cumulative distribution function $F_e$, then
we assume that $Z_e$ gets selected whenever
$Z_e \geq q_e:=F_e^{-1}(1-p_e)$ and, consequently,
the contribution
of $e$ to the objective of the prophet is
$g_e(p_e)=\int_{q_e}^{\infty} x \cdot dF_e(x)$.
More generally, when $Z_e$ does not follow a
continuous distribution, the expected value of $Z_e$
on the highest $p_e$-fraction of realizations can
be described as follows:
\ifbool{SODACamera}
{
\begin{align*}
g_e(p_e) &=
(p_e - (1-F_e(q_e)))\cdot q_e \\
& \qquad +
\int_{x\in (q_e,\infty)} x\cdot dF_e(x),
  \quad\text{ where}\\[0.5em]
q_e &= q_e(p_e) = \min\{\alpha \mid F_e(\alpha)\geq 1-p_e\}
\enspace.
\end{align*}
}
{
\begin{align*}
g_e(p_e) &=
(p_e - (1-F_e(q_e)))\cdot q_e +
\int_{x\in (q_e,\infty)} x\cdot dF_e(x),
  \quad\text{ where}\\[0.5em]
q_e &= q_e(p_e) = \min\{\alpha \mid F_e(\alpha)\geq 1-p_e\}
\enspace.
\end{align*}
}
In words, we assume that $e$ gets selected whenever
$Z_e > q_e$. Moreover, if $Z_e=q_e$, then $e$ gets
selected with probability $p_e - 1 + F_e(q_e)$.

Putting things together, the relaxation we consider
is the following.
\begin{equation}\label{eq:prophetRelax}
\max_{p\in P} \sum_{e\in N} g_e(p_e) 
\end{equation}
By the above discussion this is indeed a relaxation.
Moreover, one can easily observe that $g_e$ is
a concave functions, and thus, in many settings one
can efficiently obtain a near-optimal solution to
this relaxation using convex optimization techniques.
However, if we are only interested in proving the
existence of a prophet inequality as stated
in Theorem~\ref{thm:prophet}, we do not need an
efficient procedure to solve~\eqref{eq:prophetRelax}.

Let $x\in P$ be an optimal
solution to~\eqref{eq:prophetRelax}.
We now create an algorithm for the Bayesian online
selection problem based on the point $x\in P$
and the $c$-selectable OCRS for $P$ which exists
by assumption.
Whenever an element $e\in N$ reveals in the
Bayesian online selection problem we say that
$e$ is active if its random variable $Z_e$ realizes
within the largest $x_e$-fraction of realizations.
More formally, $e$ is active if either:
\begin{enumerate}[(i)]
\setlength\itemsep{0em}
\item $Z_e > q_e(x)$, or

\item if $Z_e=q_e(x)$ (assuming $\Pr[Z_e=q_e(x)]>0$),
we toss a coin
and declare $e$ to be active with probability
$\frac{x_e -1 + F_e(q_e(x))}{\Pr[Z_e=q_e(x)]}$.
\end{enumerate}
Let $A\subseteq N$ be the random set of active
elements. Observe that $A$ is distributed like $R(x)$,
the random subset of $N$ that contains each element
$e\in N$ with probability $x_e$ independently
of the others.
Also, by definition of active elements we have
\begin{equation}\label{eq:actContr}
\E[Z_e \cdot \characteristic_{e\in A}] = g_e(x_e)
  \qquad \forall e\in N
	\enspace.
\end{equation}
Our algorithm for the Bayesian online selection
problem applies a $c$-selectable OCRS to the
set $A$ to obtain a random set
$I\subseteq A, I\in \mathcal{F}$.
To prove the theorem, we show that the
expected value of $I$ is at least a $c$-fraction
of the optimal value of~\eqref{eq:prophetRelax},
\ie,
$\E[\sum_{e\in I} Z_e] \geq
  c \cdot \sum_{e\in N}g_e(x_e)$.

Since the OCRS is $c$-selectable, we have
\begin{equation}\label{eq:prophetCSel}
\Pr[e\in I] \geq c\cdot x_e \qquad \forall e\in N
\enspace.
\end{equation}
A key observation is that the distribution of $Z_e$
conditioned on $e\in I$ is the same as the distribution
of $Z_e$ conditioned on $e\in A$. This follows from
the fact that the OCRS does not consider the precise
value of $Z_e$, but only knows whether $e\in A$ or not.
In particular, this implies
\begin{equation}\label{eq:selToAct}
\E[Z_e \mid e\in I] = \E[Z_e \mid e\in A]
  \qquad \forall e\in N
	\enspace.
\end{equation}
Combining the above observations we get
\ifbool{SODACamera}
{
\begin{align*}
\E\left[\sum_{e\in I} Z_e\right] &=
  \sum_{e\in N} \Pr[e\in I] \cdot \E[Z_e \mid e\in I]\\
   &\geq c \cdot \sum_{e\in N} x_e \cdot \E[Z_e \mid e\in A] \\
   &= c\cdot \sum_{e\in N} \Pr[e\in A] \cdot \E[Z_e \cdot \characteristic_{e\in A} \mid e \in A] \\
   &= c\cdot \sum_{e\in N} g_e(x_e)
			\enspace,
\end{align*}
where the first inequality follows from \eqref{eq:prophetCSel}
and~\eqref{eq:selToAct}; the second equality since $\Pr[e\in A] = x_e$; and the last equality by~\eqref{eq:actContr}.
This shows that our procedure is worse by at most a factor
of $c$ compared to the value of the relaxation~\eqref{eq:prophetRelax},
which completes the proof.
}
{
\begin{align*}
\E\left[\sum_{e\in I} Z_e\right] &=
  \sum_{e\in N} \Pr[e\in I] \cdot \E[Z_e \mid e\in I]\\
   &\geq c \cdot \sum_{e\in N} x_e \cdot \E[Z_e \mid e\in A]
      && \text{(by \eqref{eq:prophetCSel}
             and~\eqref{eq:selToAct})}\\
   &= c\cdot \sum_{e\in N} \Pr[e\in A] \cdot \E[Z_e \cdot \characteristic_{e\in A} \mid e \in A]
			&& \text{(since $\Pr[e\in A] = x_e$)}\\
   &= c\cdot \sum_{e\in N} g_e(x_e)
      && \text{(by~\eqref{eq:actContr})}
			\enspace.
\end{align*}
The last inequality shows that our procedure is worse by at most a factor
of $c$ compared to the value of the relaxation~\eqref{eq:prophetRelax},
which completes the proof.
}

\subsection{Oblivious posted pricing mechanisms}

\subsubsection*{Proof of Theorem~\ref{thm:copm}\ifbool{SODACamera}{.}{}}

The proof for Theorem~\ref{thm:copm} goes along the same
lines as the proof of Theorem~\ref{thm:prophet}
presented in Section~\ref{subsec:appProphet}, but
uses a different relaxation.
The relaxation we employ is the same as the one
used by Yan~\cite{yan_2010_mechanism}. For completeness
and ease of understanding we replicate some of the
arguments in~\cite{yan_2010_mechanism} and refer
to the excellent discussion of this relaxation
in~\cite{yan_2010_mechanism}
for more details about it.

Consider the random set of agents
$I^*\subseteq N, I^*\in \mathcal{F}$ served by
Myerson's mechanism, \ie, an optimal truthful
mechanism, and let $q_e^*=\Pr[e\in I^*]$ be
the probability that $e$ gets served.
Since only feasible subsets of agents can be
served, we have $q^* \in P$, because 
$P$ is a relaxation of $\mathcal{F}$.

For this fixed $q^*$ we can now define
independent mechanism design problems for the
different agents as follows. For each $e\in N$,
we are interested in finding a price distribution
(for the price offered to $e$), that maximizes
the expected revenue under the constraint
that $e$ gets served with probability
equal to $q^*_e$.
Based on results by Myerson~\cite{myerson_1981_optimal},
it follows that the optimal price distribution can be
chosen to be a two-price distribution,
which can be determined through
a well-known technique in mechanism design
known as \emph{ironing} (we refer to~\cite{yan_2010_mechanism}
for details).
We denote by 
$R_e(q^*_e)$ the expected revenue of this optimal
distribution, which
can be shown to be concave in $q^*_e$, 
and by $\mathcal{D}_e(q^*_e)$ the distribution
itself.
Since the family of these independent mechanism
design problems for the agents is less constrained
than the original BSMD, in which we also had to make sure
that the set of all served agents is in $\mathcal{F}$,
we have that $\sum_{e\in N} R_e(q^*_e)$
is an upper bound to the expected
revenue of an optimal mechanism for the  
original BSMD (we refer to~\cite{yan_2010_mechanism} for
a formal proof).
Hence, the following is a convex relaxation of
the original BSMD.
\begin{equation}\label{eq:relaxBSMD}
  \max\left\{\sum_{e\in N}R_e(q_e)
    \;\middle\vert\; q\in P\right\}
\end{equation}

The COPM we construct needs probabilities $x\in P$ such
that $\sum_{e\in N} R_e(x_e)$ is an upper bound on the
revenue of Myerson's mechanism.
By the above discussion,
this holds for $x=q^*$, or for $x$ being an optimal
solution to~\eqref{eq:relaxBSMD}.
To make this step constructive one can follow, for example,
the sampling-based approach of Chawla et
al.~\cite{chawla_2010_multi-parameter_LONG};
they estimate the probabilities $q^*$ by
running Myerson's mechanism sufficiently many times.
Alternatively, one could use convex optimization
techniques to optimize the relaxation~\eqref{eq:relaxBSMD}.%
\footnote{We highlight that even if only an approximately optimal
$x\in P$ can be obtained, \ie, one such that
$\alpha \cdot \sum_{e\in N}R_e(x_e)$ upper bounds the
optimal revenue for some $\alpha >1$, then all of
what follows still goes through simply with an additional
loss of a factor $\alpha$. This will lead to
a COPM that is at most a factor of $\alpha\cdot c$
worse than Myerson's mechanism.}
Our COPM is randomized and defined by the following
randomization over tuples $(p, \mathcal{F}')$. The price
vector $p\in\mathbb{R}_{\geq 0}^N$ is drawn according to
the product distribution
$\bigtimes_{e\in N} \mathcal{D}_e(x_e)$, where each $p_e$ for
$e\in N$ is drawn independently according
to the two-price distribution $\mathcal{D}_e(x_e)$.
The family $\mathcal{F}'$ is chosen to be equal to
the family $\mathcal{F}_x$ of the $c$-selectable OCRS
for the point $x\in P$. Hence, if our OCRS is deterministic,
then also $\mathcal{F}'=\mathcal{F}_x$ is deterministic,
in which case the randomization of our COPM is solely on
the price vector $p$.

To prove Theorem~\ref{thm:copm} we show that
our COPM has an expected revenue of at least
$c \cdot \sum_{e\in N} R_e(x_e)$.
We call an agent $e\in N$ \emph{active} if its personal
(random) valuation $Z_e$ is at least as large as
the (random) price $p_e$. In other words, we say
that $e$ is active if it would accept the offer
presented by our COPM.
By the definition of the distributions $\mathcal{D}_e(x_e)$,
we have that each agent $e$ is active with probability
$x_e$, independently of all other agents.
Notice that an agent $e\in N$ being active does
not imply that $e$ gets selected, because a COPM
is allowed to reject an agent if feasibility
in $\mathcal{F}'$ is not maintained.
However, because our OCRS is $c$-selectable,
we have that for any $e\in N$ with probability
at least $c$, the agent $e$ can be added to the agents served
so far no matter which subset of the active agents
has been already served; furthermore, this event
is independent of whether $e$ is active itself.
The expected revenue that our COPM gets from agent
$e$ is therefore at least $c\cdot R_e(x_e)$, which
completes the proof.

\subsubsection*{Extension to other objectives\ifbool{SODACamera}{.}{}}

Finally, we notice that, as highlighted
by Yan~\cite{yan_2010_mechanism}, the used
relaxation can easily be extended to a much larger
class of objective functions that are decomposable
with respect to the agents. More precisely, these
are objectives of the form 
$\E[\sum_{e\in N} \characteristic_{Z_e\geq p_e}\cdot
g_e(Z_e,p_e)]$, where for $e\in N$,
$g_e$ is a function of the
(random) valuation $Z_e$ of $e$ and the 
price $p_e$ offered to $e$.
In particular, the maximization of revenue which
we discussed above corresponds to
$g_e(Z_e, p_e) = p_e$. Similarly, one can deal with
welfare maximization or surplus maximization by
defining $g_e(Z_e,p_e)=Z_e$ and $g_e(Z_e,p_e)=Z_e-p_e$,
respectively.
The above reasoning why our COPM is at most a factor
of $c$ worse than the optimal truthful mechanism
extends to such objectives without modifications.

\subsection{Stochastic probing\ifbool{SODACamera}{.}{}}

We begin with the proof of Theorem~\ref{thm:ocrsProbing}. Consider the following relaxation. We use $w$ in the relaxation to denote the natural extension of $w$ to vectors (formally, $w(x) = \sum_{e \in N} w(e) \cdot x_e$). Additionally, the relaxation uses the expression $p \circ x$ to denote the element-wise multiplication of the probabilities vector $p$ and the variables vector $x$. Clearly this relaxation can be solved efficiently when there is a separation oracle for the polytopes $P_{in}$ and $P_{out}$.

\[ \begin{array}{lll}
	(R1) & \max & w(p \circ x) \\
	&& p \circ x \in P_{in} \\
	&& x \in P_{out} \\
	&& x \in [0, 1]^N \\
\end{array} \]

The following lemma proves an important property of $(R1)$. The proof of this lemma is based on the observation that one feasible solution for $(R1)$ is the vector $x \in [0, 1]^N$ in which $x_e$ is equal to the marginal probability that the optimal algorithm probes element $e$.

\begin{lemma}[Claim~3.1 of~\cite{gupta_2013_stochastic}] \label{lem:optimal_bound}
The optimal value of $(R1)$ upper bounds the the expected performance of the optimal algorithm for the weighted stochastic probing problem.
\end{lemma}

Let $x^*$ be an optimal solution for $(R1)$. By Lemma~\ref{lem:optimal_bound}, to prove Theorem~\ref{thm:ocrsProbing} we only need to show an algorithm for the weighted stochastic probing problem which finds a solution of expected weight at least $b(c_{in} \cdot c_{out}) \cdot w(p \circ x^*)$. Our algorithm for the weighted stochastic probing problem is given as Algorithm~\ref{alg:probing}.

\ifbool{SODACamera}{\begin{algorithm2e}}{\begin{algorithm}}
\caption{\textsf{Probing Algorithm}} 
\label{alg:probing}
\DontPrintSemicolon
Let $A_{out}$ be a random set distributed like $R(bx^*)$.\\
Let $\hat{\cF}_{out} \subseteq \cF_{out}$ be an instantiation of the random family $\cF_{\pi_{out}, bx^*}$.\\
Let $\hat{\cF}_{in} \subseteq \cF_{in}$ be an instantiation of the random family $\cF_{\pi_{in}, p \circ (bx^*)}$.\\

\BlankLine

Let $Q, S \gets \varnothing$.\\
\For{every element $e \in N$ in the order chosen by the adversary}
{
	\If{$e \in A_{out}$, $S \cup \{e\} \in \hat{\cF}_{in}$ and $Q \cup \{e\} \in \hat{\cF}_{out}$ \label{line:conditions}}{
		Add $e$ to $Q$ and probe $e$.\\
		\lIf{$e$ is active}{Add $e$ to $S$.}
	}
}
\Return{S}.
\ifbool{SODACamera}{\end{algorithm2e}}{\end{algorithm}}

The next observation shows that Algorithm~\ref{alg:probing} is a legal algorithm for the weighted stochastic probing problem. 

\begin{observation}
The following always hold when Algorithm~\ref{alg:probing} terminates:
\begin{compactitem}
\item The set $Q$ of probed elements is in $\hat{\cF}_{out} \subseteq \cF_{out}$.
\item The set $S$ of selected elements is in $\hat{\cF}_{in} \subseteq \cF_{in}$.
\item The set $S$ contains exactly the active elements of $Q$.
\end{compactitem}
\end{observation}

In the rest of the section we use notation and results introduced in Section~\ref{sec:selectability_to_approximation}. Let $\bar{\pi}_{in}$ and $\bar{\pi}_{out}$ denote the characteristic CRSs of the greedy OCRSs $\pi_{in}$ and $\pi_{out}$, respectively. We would like to use $\bar{\pi}_{in}$ and $\bar{\pi}_{out}$ to lower bound the value of a subset of $S$, and through that subset also the value of $S$. To achieve that objective we first need to describe the said subset of $S$ as an expression of $\bar{\pi}_{in}$ and $\bar{\pi}_{out}$.

Let $A_{in}$ be the set of elements that belong to $A_{out}$ and are also active. One can observe that $A_{in}$ is distributed like $R(p \circ x)$. Additionally, let us couple the randomness of Algorithm~\ref{alg:probing} and the CRSs $\bar{\pi}_{in}$ and $\bar{\pi}_{out}$ as follows: we use the same instantiations for the random families $\cF_{\pi_{in}, p \circ (bx^*)}$ and $\cF_{\pi_{out}, bx^*}$ in both Algorithm~\ref{alg:probing} and the definitions of $\bar{\pi}_{in}$ and $\bar{\pi}_{out}$, respectively. Notice that this coupling implies:
\ifbool{SODACamera}
{
\begin{multline*}
	\bar{\pi}_{in}(A_{in})
	=
	\{e \in A_{in} \mid I \cup \{e\} \in \hat{\cF}_{in}\\ \forall I \subseteq A_{in}, I \in \hat{\cF}_{in}\}
\end{multline*}
}
{
\[
	\bar{\pi}_{in}(A_{in})
	=
	\{e \in A_{in} \mid I \cup \{e\} \in \hat{\cF}_{in}\;\; \forall I \subseteq A_{in}, I \in \hat{\cF}_{in}\}
\]
}
and
\ifbool{SODACamera}
{
\begin{multline*}
	\bar{\pi}_{out}(A_{out})
	=
	\{e \in A_{out} \mid I \cup \{e\} \in \hat{\cF}_{out}\\ \forall I \subseteq A_{out}, I \in \hat{\cF}_{out}\}
	\enspace.
\end{multline*}
}
{
\[
	\bar{\pi}_{out}(A_{out})
	=
	\{e \in A_{out} \mid I \cup \{e\} \in \hat{\cF}_{out}\;\; \forall I \subseteq A_{out}, I \in \hat{\cF}_{out}\}
	\enspace.
\]
}

\begin{observation} \label{obs:subset_pi}
$\bar{\pi}_{in}(A_{in}) \cap \bar{\pi}_{out}(A_{out}) \subseteq S$.
\end{observation}
\begin{proof}
Consider an element $e \in \bar{\pi}_{in}(A_{in}) \cap \bar{\pi}_{out}(A_{out})$. Since $e \in \bar{\pi}_{in}(A_{in}) \subseteq A_{in}$, $e$ must be both active and in $A_{out}$. Let $S_e$ and $Q_e$ denote the sets $S$ and $Q$ immediately before $e$ is processed by Algorithm~\ref{alg:probing}. Clearly $S_e \subseteq A_{in}$ and $S_e \in \hat{\cF}_{in}$. Together with the fact that $e \in \bar{\pi}_{in}(A_{in})$, these observations imply that $S_e \cup \{e\} \in \hat{\cF}_{in}$. An analogous arguments shows also that $Q_e \cup \{e\} \in \hat{\cF}_{out}$. Hence, we proved that $e$ is an active element obeying all the conditions on Line~\ref{line:conditions} when processed by Algorithm~\ref{alg:probing}, and therefore, $e$ is added to $S$ by the algorithm.
\end{proof}

We defined $A_{out}$ as a random set distributed like $R(bx^*)$ and $A_{in}$ as the intersection of $A_{out}$ and the set of active elements. However, for the purpose of analyzing the value of the set $\bar{\pi}_{in}(A_{in}) \cap \bar{\pi}_{out}(A_{out})$ we can assume any construction procedure that results in the same joint distribution of $A_{in}$ and $A_{out}$. The following is a convenient construction that we are going to use from this point on: the set $A_{in}$ is a random set distributed like $R(p \circ (bx^*))$. The set $A_{out}$ is calculated by starting with $A_{in}$ and adding to it every element $e \not \in A_{in}$ with probability $bx^*_e(1 - p_e)/(1-bp_ex^*_e)$, independently. Notice that this construction indeed produces the same joint distribution of the sets $A_{in}$ and $A_{out}$ as the original construction. For ease of notation, let us denote by $z$ a vector in $[0, 1]^N$ defined by: $z_e = bx^*_e(1 - p_e)/(1-bp_ex^*_e)$ for every $e \in N$. Using this notation we get that $A_{out}$ has the same distribution as $A_{in} \cup R(z)$.

The new construction of $A_{out}$ implies that $A_{out}$ is a random function of $A_{in}$, and so is the expression $\bar{\pi}_{in}(A_{in}) \cap \bar{\pi}_{out}(A_{out})$. Thus, we can define a new CRS $\bar{\pi}$ for $P_{in}$ by the equality $\bar{\pi}(A_{in}) = \bar{\pi}_{in}(A_{in}) \cap \bar{\pi}_{out}(A_{out})$. Notice that $\bar{\pi}$ is a true CRS for $P_{in}$ in the sense that it always outputs a set in $\cF_{in}$ since $\bar{\pi}_{in}(A_{in})$ is guaranteed to be in $\cF_{in}$. Let us now study the properties of $\bar{\pi}$.

\begin{lemma} \label{lem:monotone}
The CRS $\bar{\pi}$ is monotone.
\end{lemma}
\begin{proof}
We need to show that every element $e \in N$ and two sets $e \in T_1 \subseteq T_2 \subseteq N$ obey the inequality:
\ifbool{SODACamera}
{
\begin{align*}
	\Pr&[e \in \bar{\pi}(A_{in}) \mid A_{in} = T_1] \\
	&=
	\Pr[e \in \bar{\pi}_{in}(A_{in}) \cap \bar{\pi}_{out}(A_{out}) \mid A_{in} = T_1]\\
	& \geq{} 
	\Pr[e \in \bar{\pi}_{in}(A_{in}) \cap \bar{\pi}_{out}(A_{out}) \mid A_{in} = T_2] \\
	& =
	\Pr[e \in \bar{\pi}(A_{in}) \mid A_{in} = T_2]
	\enspace.
\end{align*}
}
{
\begin{align*}
	&
	\Pr[e \in \bar{\pi}(A_{in}) \mid A_{in} = T_1]
	=
	\Pr[e \in \bar{\pi}_{in}(A_{in}) \cap \bar{\pi}_{out}(A_{out}) \mid A_{in} = T_1]\\
	\geq{} &
	\Pr[e \in \bar{\pi}_{in}(A_{in}) \cap \bar{\pi}_{out}(A_{out}) \mid A_{in} = T_2]
	=
	\Pr[e \in \bar{\pi}(A_{in}) \mid A_{in} = T_2]
	\enspace.
\end{align*}
}
This is true since:
\ifbool{SODACamera}
{
\begin{align*}
	\Pr&[e \in \bar{\pi}_{in}(A_{in}) \cap \bar{\pi}_{out}(A_{out}) \mid A_{in} = T_1] \\
	& =
	\Pr[e \in \bar{\pi}_{in}(T_1) \cap \bar{\pi}_{out}(T_1 \cup R(z))]\\
	& ={} 
	\Pr[e \in \bar{\pi}_{in}(T_1)] \cdot \Pr[e \in \bar{\pi}_{out}(T_1 \cup R(z))] \\
	& \geq
	\Pr[e \in \bar{\pi}_{in}(T_2)] \cdot \Pr[e \in \bar{\pi}_{out}(T_2 \cup R(z))]\\
	& =
	\Pr[e \in \bar{\pi}_{in}(T_2) \cap \bar{\pi}_{out}(T_2 \cup R(z))] \\
	& =
	\Pr[e \in \bar{\pi}_{in}(A_{in}) \cap \bar{\pi}_{out}(A_{out}) \mid A_{in} = T_2]
	\enspace,
\end{align*}
}
{
\begin{align*}
	&
	\Pr[e \in \bar{\pi}_{in}(A_{in}) \cap \bar{\pi}_{out}(A_{out}) \mid A_{in} = T_1]
	=
	\Pr[e \in \bar{\pi}_{in}(T_1) \cap \bar{\pi}_{out}(T_1 \cup R(z))]\\
	={} &
	\Pr[e \in \bar{\pi}_{in}(T_1)] \cdot \Pr[e \in \bar{\pi}_{out}(T_1 \cup R(z))]
	\geq
	\Pr[e \in \bar{\pi}_{in}(T_2)] \cdot \Pr[e \in \bar{\pi}_{out}(T_2 \cup R(z))]\\
	=&
	\Pr[e \in \bar{\pi}_{in}(T_2) \cap \bar{\pi}_{out}(T_2 \cup R(z))]
	=
	\Pr[e \in \bar{\pi}_{in}(A_{in}) \cap \bar{\pi}_{out}(A_{out}) \mid A_{in} = T_2]
	\enspace,
\end{align*}
}
where the inequality follows since both $\bar{\pi}_{in}$ and $\bar{\pi}_{out}$ are monotone by Lemma~\ref{lem:characteristic}.
\end{proof}

The following lemma shows that $\bar{\pi}$ obeys a weak variant of balanceness.

\begin{definition}
A CRS $\pi$ for a polytope $P$ is $(x, c)$-balanced for a vector $x$ and $c \in [0, 1]$ if $\Pr[e \in \pi(R(x)) \mid e \in R(x)] \geq c$ for every element $e \in N$ having $x_e > 0$.
\end{definition}

\begin{lemma} \label{lem:balanceness}
The CRS $\bar{\pi}$ is $(p \circ (bx^*), c_{in} \cdot c_{out})$-balanced.
\end{lemma}
\begin{proof}
Since $A_{in}$ is distributed like $R(p \circ (bx^*))$, we need to show that $\Pr[e \in \bar{\pi}_{in}(A_{in}) \cap \bar{\pi}_{out}(A_{out}) \mid e \in A_{in}] \geq c_1 \cdot c_2$ holds for every element $e \in N$ having $bp_ex^*_e > 0$. 

Since $\bar{\pi}_{in}$ and $\bar{\pi}_{out}$ are monotone by Lemma~\ref{lem:characteristic}, we get that both $\Pr[e \in \bar{\pi}_{in}(T)]$ and $\Pr[e \in \bar{\pi}_{out}(T \cup R(z))]$ are decreasing functions of $T$ as long as $e$ is in $T$. Thus, by the FKG inequality:
\ifbool{SODACamera}
{\allowdisplaybreaks
\begin{align*}
	\Pr& [e \in \bar{\pi}_{in}(A_{in}) \cap \bar{\pi}_{out}(A_{out}) \mid e\in A_{in}] \\
	& =
	\Pr[e \in \bar{\pi}_{in}(A_{in}) \cap \bar{\pi}_{out}(A_{in} \cup R(z)) \mid e \in A_{in}]\\
	& \geq{} 
	\Pr[e \in \bar{\pi}_{in}(A_{in}) \mid e \in A_{in}]   \\
  & \qquad \quad \cdot\Pr[e \in \bar{\pi}_{out}(A_{in} \cup R(z)) \mid e \in A_{in}]\\
	& =
	\Pr[e \in \bar{\pi}_{in}(A_{in}) \mid e \in A_{in}] \\
  & \qquad\quad \cdot \Pr[e \in \bar{\pi}_{out}(A_{out}) \mid e \in A_{out}]
	\enspace,
\end{align*}
}
{
\begin{align*}
	\Pr[e \in \bar{\pi}_{in}(A_{in}) \cap \bar{\pi}_{out}(A_{out}) \mid e&\in A_{in}]
	=
	\Pr[e \in \bar{\pi}_{in}(A_{in}) \cap \bar{\pi}_{out}(A_{in} \cup R(z)) \mid e \in A_{in}]\\
	\geq{} &
	\Pr[e \in \bar{\pi}_{in}(A_{in}) \mid e \in A_{in}] \cdot \Pr[e \in \bar{\pi}_{out}(A_{in} \cup R(z)) \mid e \in A_{in}]\\
	={} &
	\Pr[e \in \bar{\pi}_{in}(A_{in}) \mid e \in A_{in}] \cdot \Pr[e \in \bar{\pi}_{out}(A_{out}) \mid e \in A_{out}]
	\enspace,
\end{align*}
}
where the last equality uses the fact that the membership of every element in the sets $A_{in}$ and $R(z)$ is independent from the membership of other elements in these sets.

The set $A_{in}$ is distributed like $R(p \circ (bx^*))$ and the vector $p \circ (bx^*)$ is inside the polytope $bP_{in}$. Thus, since $\bar{\pi}_{in}$ is $(b, c_{in})$-balanced by Lemma~\ref{lem:characteristic}, we get: $\Pr[e \in \bar{\pi}_{in}(A_{in}) \mid e \in A_{in}] \geq c_{in}$. Similarly, the set $A_{out}$ is distributed like $R(bx^*)$ and the vector $bx^*$ is inside the polytope $bP_{out}$. Thus, since $\bar{\pi}_{out}$ is $(b, c_{out})$-balanced by Lemma~\ref{lem:characteristic}, we get: $\Pr[e \in \bar{\pi}_{in}(A_{out}) \mid e \in A_{out}] \geq c_{out}$. The lemma now follows by combining the above inequalities.
\end{proof}

In the following corollary we use $w$ to denote its natural extension to sets, \ie, $w(T) = \sum_{e \in T} w(e)$.

\begin{corollary}
$\bE[w(S)] \geq b(c_{in} \cdot c_{out}) \cdot w(p \circ x^*)$.
\end{corollary}
\begin{proof}
%Observe that the function $w\colon 2^N \to \bR_{\geq 0}$ is non-negative and linear (and thus, also submodular). Moreover, the multilinear extension of $w$ is its natural extension to vectors. Next, 
By Lemma~\ref{lem:balanceness}, linearity of the expectation and the observation that $A_{in}$ is distributed like $R(p \circ (bx^*))$, we get:
\ifbool{SODACamera}
{
\begin{align*}
	\bE  & [w(\bar{\pi}(A_{in}))]
	 ={} 
	\sum_{e \in N} w(e) \cdot \Pr[e \in \bar{\pi}(R(p \circ (bx^*)))]\\
	& ={} 
	\sum_{e \in N} w(e) \cdot \Pr[e \in R(p \circ (bx^*))] \\
  & \qquad \quad \cdot \Pr[e \in \bar{\pi}(R(p \circ (bx^*))) \mid e \in R(p \circ (bx^*))]\\
	& \geq{} 
	\sum_{e \in N} w(e) \cdot (p_e \circ (bx^*_e)) \cdot (c_{in} \cdot c_{out}) \\
	%=
	%(c_{in} \cdot c_{out}) \cdot w(p \circ (bx^*))
	& =
	b(c_{in} \cdot c_{out}) \cdot w(p \circ x^*)
	\enspace.
\end{align*}
}
{
\begin{align*}
	\bE[w(\bar{\pi}(A_{in}))]
	={} &
	\sum_{e \in N} w(e) \cdot \Pr[e \in \bar{\pi}(R(p \circ (bx^*)))]\\
	={} &
	\sum_{e \in N} w(e) \cdot \Pr[e \in R(p \circ (bx^*))] \cdot \Pr[e \in \bar{\pi}(R(p \circ (bx^*))) \mid e \in R(p \circ (bx^*))]\\
	\geq{} &
	\sum_{e \in N} w(e) \cdot (p_e \circ (bx^*_e)) \cdot (c_{in} \cdot c_{out})
	%=
	%(c_{in} \cdot c_{out}) \cdot w(p \circ (bx^*))
	=
	b(c_{in} \cdot c_{out}) \cdot w(p \circ x^*)
	\enspace.
\end{align*}
}
Since all the weights are non-negative and $\bar{\pi}(A_{in}) = \bar{\pi}_{in}(A_{in}) \cap \bar{\pi}_{out}(A_{out}) \subseteq S$ by Observation~\ref{obs:subset_pi}, the last inequality implies $\bE[w(S)] \geq \bE[w(\bar{\pi}(A_{in}))] \geq b(c_{in} \cdot c_{out}) \cdot w(p \circ x^*)$.
\end{proof}

Following the above discussion, Theorem~\ref{thm:ocrsProbing} is implied by the last corollary. We can now prove Theorem~\ref{thm:ocrsProbDeadline} as a direct consequence of Theorem~\ref{thm:ocrsProbing}.

\begin{proof}[Proof of Theorem~\ref{thm:ocrsProbDeadline}]
Consider a down-closed set $\cF \subseteq 2^N$ containing every set $T \subseteq N$ if and only if all the elements of $T$ can be queried while respecting the deadlines. Formally,
\[
	\cF
	=
	\Big\{T \subseteq 2^N \;\Big\vert\; \forall_{1 \leq d \leq |N|} \; |\{e \in T \mid d_e \leq d\}| \leq d\Big\}
	\enspace.
\]
By definition $(N, \cF)$ is a laminar matroid, thus, by Theorem~\ref{thm:direct_OCRS} there exists a $(b, 1 - b)$-selectable greedy OCRS for its matroid polytope $P_\cF$. Together with the existence of $\pi_{out}$ we get, by Theorem~\ref{thm:combineOCRSs}, a $(b, (1 - b)c_{out})$-selectable greedy OCRS for the polytope $P'_{out} = P_{out} \cap P_\cF$. Notice that this polytope is a relaxation of the down-closed family $\cF'_{out} = \cF_{out} \cap \cF$. Moreover, $P'_{out}$ has a separation oracle whenever $P_{out}$ has such an oracle.

For convenience, let us use $(\cF_1, \cF_2)$-probing as a shorthand for the weighted stochastic probing problem with $\cF_1$ and $\cF_2$ as the inner and outer constraints, respectively. Consider the best algorithm for $(\cF_{in}, \cF_{out})$-probing with deadlines. Since this algorithm probes with respect to the deadlines, the set of elements it probes must be in $\cF$. Hence, the same algorithm is also an algorithm for $(\cF_{in}, \cF'_{out})$-probing. Thus, by Theorem~\ref{thm:ocrsProbing} we have an algorithm $ALG$ for $(\cF_{in}, \cF'_{out})$-probing whose approximation ratio is $b(1 - b) \cdot c_{in} \cdot c_{out}$ compared to the best algorithm for $(\cF_{in}, \cF_{out})$-probing with deadlines. Moreover, the approximation ratio of $ALG$ holds regardless of the order in which $ALG$ can probe elements.

The algorithm we suggest for $(\cF_{in}, \cF_{out})$-probing with deadlines is $ALG$ when we allow it to probe elements in increasing deadlines order. We have already proved that $ALG$ has the approximation ratio guaranteed by the theorem, so we only need to explain why does it respect the deadlines. Assume towards a contradiction that $ALG$ probes element $e$ after time $d_e$. This means that $ALG$ probes a set $T$ of at least $d_e$ elements before it probes $e$. However, since $ALG$ can probe elements only in increasing deadlines order, all the elements of $T \cup \{e\}$ have a deadline of at most $d_e$. The last observation implies that $T \cup \{e\} \not \in \cF$, which contradicts the fact that the set of elements probed by $ALG$ is always in $\cF'_{out} \subseteq \cF$.
\end{proof}

The rest of this section is devoted to proving Theorem~\ref{thm:ocrsProbSubm}. Let $f\colon 2^N \to \bR_{\geq 0}$ be the non-negative monotone submodular objective function of the problem, \ie, the value of an output $S$ of the probing algorithm is $f(S)$. We need to introduce an extension of $f$ to $[0, 1]^N$ studied by~\cite{calinescu_2007_maximizing}.
\ifbool{SODACamera}
{
\begin{multline*}
	f^+(x) = \max\left\{\sum_{T \subseteq N}  \alpha_T \cdot f(T) ~\middle|~ \sum_{T \subseteq N} \alpha_T \leq 1,\right. \\ \left. \alpha_T \geq 0 \text{ and } \forall_{e \in N} \; \sum_{e \in T \subseteq N} \alpha_T \leq x_e\right\}
	\enspace.
\end{multline*}
}
{
\[
	f^+(x) = \max\left\{\sum_{T \subseteq N}  \alpha_T \cdot f(T) ~\middle|~ \sum_{T \subseteq N} \alpha_T \leq 1, \alpha_T \geq 0 \text{ and } \forall_{e \in N} \; \sum_{e \in T \subseteq N} \alpha_T \leq x_e\right\}
	\enspace.
\]
}

Intuitively, $f^+$ is equal to the largest possible expected value of $f$ over a distribution of sets in which every element $e$ appears with a marginal probability at most $x_e$. Using this extension, we can now introduce a variant of the relaxation $(R1)$ that works for monotone submodular objectives.

\[ \begin{array}{lll}
	(R2) & \max & f^+(p \circ x) \\
	&& p \circ x \in P_{in} \\
	&& x \in P_{out} \\
	&& x \in [0, 1]^N \\
\end{array} \]

Let us explain why we use the extension $f^+$ in $(R2)$ instead of the simpler multilinear extension. An algorithm for the submodular stochastic probing may choose the next element to probe based on the set of elements previously probed and the results of these probes. Thus, the membership of elements in the solution produced by the algorithm is not independent, and this is captured by $f^+$. Using this intuition, the work of~\cite{adamczyk_2014_submodular} implies the following counterpart of Lemma~\ref{lem:optimal_bound}.

\begin{lemma} \label{lem:optimal_bound_submodular}
The optimal value of $(R2)$ upper bounds the the expected performance of the optimal algorithm for the submodular stochastic probing problem.
\end{lemma}

Additionally, \cite{adamczyk_2014_submodular} also shows that a variant of the continuous greedy algorithm of~\cite{culinescu_2011_maximizing} can be used to find a point $x$ obeying all the constraints of $(R2)$ and also the inequality $F(p \circ x) \geq (1 - e^{-1}) \cdot f^+(p \circ x^*)$, where $F$ is the multilinear extension of $f$ and $x^*$ is the optimal solution for $(R2)$. The same argument can also be used to show that by stopping the continuous greedy algorithm at time $b$, instead of letting it reach time $1$, one gets a point $\tilde{x} \in [0, 1]^N$ obeying $\tilde{p} \circ x \in bP_{in}$, $\tilde{x} \in bP_{out}$ and $F(p \circ \tilde{x}) \geq (1 - e^{-b} - o(1)) \cdot f^+(p \circ x^*)$.\footnote{The idea of running the continuous greedy for less time instead of scaling its output was first introduced by~\cite{feldman_2011_unified}.}

The algorithm we use to prove Theorem~\ref{thm:ocrsProbSubm} is Algorithm~\ref{alg:probing} modified by replacing the vector $bx^*$ with $\tilde{x}$. Additionally, we use the same definition of the CRS $\bar{\pi}$ given above, and observe that Lemmata~\ref{lem:monotone} and \ref{lem:balanceness} still apply to this CRS (with $\tilde{x}$ replacing $bx^*$ in the appropriate places). To analyze the output of this CRS we observe that the proof of~\cite{chekuri_2014_submodular} for Lemma~\ref{lem:offline_crs} in fact proves the following stronger version of the lemma. Note that this version strictly generalizes Lemma~\ref{lem:offline_crs} since being $(b, c)$-balanced is equivalent to being $(x, c)$-balanced for every vector $x \in bP$.

\begin{lemma} \label{lem:offline_crs_stronger}
For every given non-negative submodular function $f\colon 2^N \to \bR_{\geq 0}$, there exists a function $\eta_f \colon 2^N \to 2^N$ that always returns a subset of its argument (\ie, $\eta_f(S) \subseteq S$ for every $S \subseteq N$) having the following property. For every input vector $x \in [0, 1]^N$ and monotone $(x, c)$-balanced CRS $\pi$:
\[
	\bE[f(\eta_f(\pi(R(x))))]
	\geq
	c \cdot F(x)
	\enspace,
\]
where $F(x)$ is the multilinear extension of $f$.
\end{lemma}

We are now ready to prove the next lemma, which together with Lemma~\ref{lem:optimal_bound_submodular}, proves Theorem~\ref{thm:ocrsProbSubm}.
\begin{lemma}
$\bE[f(S)] \geq (c_{in} \cdot c_{out}) \cdot F(p \circ \tilde{x})$, where $S$ is the output set of the modified Algorithm~\ref{alg:probing}.
\end{lemma}
\begin{proof}
Let $S'$ be the set $\eta_f(\pi(A_{in})) \subseteq \pi(A_{in}) = \bar{\pi}_{in}(A_{in}) \cap \bar{\pi}_{out}(A_{out}) \subseteq S$. By combining Lemmata~\ref{lem:monotone}, \ref{lem:balanceness} and~\ref{lem:offline_crs_stronger} and using the observation that $A_{in}$ is distributed like $R(p \circ \tilde{x})$, we get:
\[
	\bE[f(S')]
	\geq
	(c_{in} \cdot c_{out}) \cdot F(p \circ \tilde{x})
	\enspace.
\]
Using the monotonicity of $f$ we now get: $\bE[f(S)] \geq \bE[f(S')] \geq (c_{in} \cdot c_{out}) \cdot F(p \circ \tilde{x})$.
\end{proof}

%% file: matroidConcepts.tex
\section{Definitions of Matroidal Concepts} \label{app:matroid_definitions}

Recall that a matroid $M=(N,\mathcal{F})$ is a tuple consisting of a finite ground set $N$, and a nonempty family $\mathcal{F}\subseteq 2^N$ of subsets of
the ground set, called \emph{independent} sets, which satisfy:
\begin{enumerate}[(i)]
\setlength\itemsep{0em}
\item $I\subseteq J \in \mathcal{F} \Rightarrow I\in \mathcal{F}$,
and
\item $I,J\in \mathcal{F}, |I|>|J|$ $\Rightarrow$
$\exists e\in I\setminus J$ s.t.
$J\cup\{e\}\in \mathcal{F}$.
\end{enumerate}

The \emph{rank} of a set $S \subseteq 2^N$ is the size of a maximum cardinality independent subset of $S$. The rank function of the matroid $M$ is a function $\rank\colon 2^N \to \mathbb{Z}_{\geq 0}$ (where $\mathbb{Z}_{\geq 0}$ is the set of all non-negative integers) assigning each set its rank.
More formally,
\[
	\rank(S)
	=
  \max\{|I| \mid I\in \mathcal{F}, I\subseteq S\}
	\enspace.
\]
Observe that the rank of an independent set is equal to its size. The rank of the matroid $M$ itself is defined as $\rank(N)$. Notice that $\rank(S) \leq \rank(N)$ for every subset $S \subseteq N$. A set $S \subseteq N$ is called a \emph{base} of $M$ if it is independent and has maximum rank, \ie, $|S| = \rank(S) = \rank(N)$.

We say that an element $e \in N$ is \emph{spanned} by a set $S \subseteq N$ if adding $e$ to $S$ does not increase the rank of $S$. On the other hand, the \emph{span} of a subset $S\subseteq N$ is the set of elements that are spanned by it. More formally, the span of a set $S$ is $\spn(S) = \{e\in N \mid \rank(S+e) = \rank(S)\}$.

Given a subset $N' \subseteq N$, the restriction of $M$ to $N'$, denoted by $M|_{N'}$, is the matroid obtained from $M$ by keeping only the elements of $N'$. Formally, $M|_{N'}$ is the matroid $(N', \mathcal{F} \cap 2^{N'})$. On the other hand, contracting $N'$ in $M$ results in another matroid, denoted by $M/N'$, over the ground set $N \setminus N'$. A set is independent in $M/N'$ if and only if adding a base of $N'$ to it results in an independent set of $M$.
It turns out that this definition is independent of the base that is chosen for $N'$.
Formally, $M/N'$ is the matroid $(N \setminus N', \mathcal{F}')$, where:
\[
  \mathcal{F}'
	=
	\{S \subseteq N \setminus N' \mid \rank(S \cup N') = |S| + \rank(N')\}
	\enspace.
\]